\documentclass[12pt,english]{article}
\pdfoutput=1

\usepackage{ae,aecompl}
\usepackage[T1]{fontenc}
\usepackage[latin9]{inputenc}
\usepackage{geometry}
\usepackage{amsmath}
\usepackage{amsthm}
\usepackage{amssymb}
\usepackage{graphicx}
\usepackage{setspace}
\usepackage{esint}
\usepackage[numbers]{natbib}
\usepackage{calc}
\usepackage{fixltx2e}
\usepackage{multicol}
\usepackage[normalem]{ulem}
\usepackage{aecompl}
\usepackage{array}
\usepackage{float}
\usepackage{multirow}
\usepackage{epsfig}
\usepackage{amsfonts}
\usepackage{lscape}
\usepackage{caption}
\usepackage{rotating}
\usepackage{xcolor}
\usepackage{amsfonts}
\usepackage{babel}
\usepackage{babel}
\usepackage{subfig}
\usepackage{babel}

\makeatletter
  \theoremstyle{remark}
  
\newtheorem{rem}{\protect\remarkname}
\theoremstyle{plain}
\newtheorem{thm}{\protect\theoremname}
\theoremstyle{plain}
  
\newtheorem{prop}{\protect\propositionname}
\theoremstyle{plain}
  
\newtheorem{lem}{\protect\lemmaname}
\theoremstyle{theorem}
\newtheorem{assumption}{Assumption}

\floatstyle{ruled}
\numberwithin{equation}{section}
\numberwithin{figure}{section}
\theoremstyle{remark}
\providecommand{\lemmaname}{Lemma}
\providecommand{\propositionname}{Proposition}
\providecommand{\remarkname}{Remark}
\providecommand{\theoremname}{Theorem}
\@ifundefined{showcaptionsetup}{}{ \PassOptionsToPackage{caption=false}{subfig}}
\makeatother
\providecommand{\lemmaname}{Lemma}
  \providecommand{\propositionname}{Proposition}
  \providecommand{\remarkname}{Remark}
\providecommand{\theoremname}{Theorem}

\begin{document}

\date{\today.}
\title{Control Variables, Discrete Instruments, and Identification of Structural
Functions}

\author{Whitney Newey\thanks{Department of Economics, MIT, wnewey@mit.edu.}
$\quad$and $\quad$Sami Stouli\thanks{Department of Economics, University of Bristol, s.stouli@bristol.ac.uk}}
\maketitle
\begin{abstract}
Control variables provide an important means of controlling for endogeneity
in econometric models with nonseparable and/or multidimensional heterogeneity.
We allow for discrete instruments, giving identification results under
a variety of restrictions on the way the endogenous variable and the
control variables affect the outcome. We consider many structural
objects of interest, such as average or quantile treatment effects.
We illustrate our results with an empirical application to Engel curve
estimation. 
\end{abstract}


\newpage{}

\section{Introduction}

Nonseparable and/or multidimensional heterogeneity is important. It
is present in discrete choice models as in McFadden (1973) and Hausman
and Wise (1978). Multidimensional heterogeneity in demand functions
allows price and income elasticities to vary over individuals in unrestricted
ways, e.g., Hausman and Newey (2016) and Kitamura and Stoye (2018).
It allows general variation in production technologies. Treatment
effects that vary across individuals require intercept and slope heterogeneity.

Endogeneity is often a problem in these models because we are interested
in the effect of an observed choice, or treatment variable on an outcome.
Control variables provide an important means of controlling for endogeneity
with multidimensional heterogeneity. A control variable is an observed
or estimable variable that makes heterogeneity and treatment independent
when it is conditioned on. Observed covariates serve as control variables
for treatment effects (Rosenbaum and Rubin, 1983). The conditional
cumulative distribution function (CDF) of a choice variable given
an instrument can serve as a control variable in economic models (Imbens
and Newey, 2009).

Nonparametric identification of many objects of interest, such as
average or quantile treatment effects, requires a full support condition,
that the support of the control variable conditional on the treatment
variable is equal to the marginal support of the control variable.
This restriction is often not satisfied in practice; e.g., see Imbens
and Newey (2009) for Engel curves. It cannot be satisfied when instruments
are discrete. One approach to this problem is to focus on identified
sets for objects of interest, as for quantile effect in Imbens and
Newey (2009). Another approach is to consider restrictions on the
model that allow for point identification. Florens et al. (2008) did
so by showing identification when the structural function is a polynomial
in the endogenous variable and a measurable separability condition
is satisfied. Torgovitsky (2015) and D'Haultfœuille and Février (2015)
did so by showing identification for discrete instruments when the
structural disturbance is a scalar.

In this paper we give identification results under a variety of restrictions
on the way the treatment and control variables enter the control regression
of the outcome of interest on the endogenous and control variables.
The control regression functions (CRF) we consider are the conditional
mean, quantile, and (monotone transformations of) distribution functions
of the outcome given the endogenous and control variables. We give
identification results when a CRF is a linear combinations of known
functions of a treatment and control variables. We also give identification
results for partially nonparametric specifications where a CRF is
a linear combination of known functions of either the treatment or
the control variables, with coefficients that are unknown functions
of the other variable.

The partially nonparametric specifications we consider generalise
those of Florens et al. (2008) to allow for nonpolynomial functions
of endogenous variables or control variables and to consider CRFs
other than the mean. We also take a different approach to identification,
focusing here on conditional nonsingularity of second moment matrices
instead of measurable separability. These results here also generalise
the identification conditions for the baseline models considered by
Chernozhukov et al. (2017). For triangular systems with a continuous
treatment, our identification results also generalise those of Masten
and Torgovitsky (2016) to allow for known functions of control variables,
and to include quantile and distribution treatment effects. For treatment
effects with a binary or discrete treatment, the present paper contributes
to the literature (Rosenbaum and Rubin, 1983; Imbens, 2000; Wooldridge,
2004) by providing conditions based on conditional nonsingularity
for identification of average treatment effects. These results complement
those of Newey and Stouli (2018) by allowing for known functions of
control variables and by considering conditional quantile and distribution
CRFs.

A main benefit of our approach is that it allows for discrete instruments.
For triangular systems, with continuous treatment, we show identification
of average, distribution, and quantile treatment effects given sufficient
variation in the discrete instrument conditional on the endogenous
variable. These results are obtained by viewing various control regression
specifications as varying coefficient models. These results generalise
the analysis of Masten and Torgovitsky (2016) to conditional distribution
and quantile effects, and to known functions of control variables.

These results provide an alternative approach to identifying objects
of interest in nonseparable models with discrete instruments. Instead
of restricting the dimension of the heterogeneity to obtain identification
with discrete instruments as done in Torgovitsky (2015) and D'Haultfœuille
and Février (2015), we can allow for multidimensional heterogeneity
but restrict the way the treatment or controls affect the outcome.

These results provide an alternative approach to identifying treatment
effects with a finite number of treatment regimes. Here the CRF depends
on treatment only through the (known) vector of dummy variables for
each regime. Nonsingularity of the conditional second moment matrix
provides a relatively simple and general condition for identification
of treatment effects. If restrictions are placed on the way the control
variables affect the CRF then the conditional nonsingularity condition
can be weakened. For example for a binary treatment regime (i.e., treated
or not) we can allow for the propensity score to be bounded away from
zero and one only on a subset of control variables values.

We illustrate our results using an empirical application to Engel
curves estimation using British expenditure survey data. We find that
estimates of average, distributional and quantile treatment effects
of total expenditure on food and leisure expenditure are not very
sensitive to discretisation of the income instruments. We find that
as we ``coarsen'' the instrument by only using knowledge of income
intervals the structural estimates do not change much until the instrument
is very coarse. Thus, in this empirical example we find that one can
obtain good structural estimates even with discrete instruments.

In Section 2 we introduce the parametric models we consider. In Section
3 we give identification results. In Section 4 we extend these results
to partially nonparametric models that allow for nonparametric components.
Section 5 reports the results of an empirical application to Engel
curve estimation.

%The proofs of the main results are given in the Appendix.%Section 5 concludes.

%\bigskip{}

%\textbf{Notation.} %For random variables $A$ and $B$, we use calligraphic
%letters $\mathcal{A}$ and $\mathcal{B}$ to denote their marginal
%supports, and $\mathcal{AB}$ to denote their joint support. Calligraphic
%letters with a subscript denote conditional supports, for instance
%the support of $A$ conditional on $B=b\in\mathcal{B}$ is denoted
%by $\mathcal{A}_{b}:=\{a\in\mathcal{A}\,:\,f_{A\mid B}(a\mid b)\geq\delta_{b},\,b\in\mathcal{B}\}$
%for some $\delta_{b}>0$, where $f_{A\mid B}(\cdot\mid b)$ denotes
%the conditional density function of $A$ given $B=b$. %Similarly,
%we denote the cumulative distribution and quantile functions of $A$
%conditional on $B=b$ by $F_{A\mid B}(\cdot\mid b)$ and $Q_{A\mid B}(\cdot\mid b)$,
%respectively. Finally, we denote the cardinality of a set $\mathcal{A}$
%by $\left|\mathcal{A}\right|$.

\section{Parametric Modelling of Control Regressions}

Let $Y$ denote an outcome variable of interest and $X$ an endogenous
treatment with supports denoted by $\mathcal{Y}$ and $\mathcal{X}$,
respectively. For $\varepsilon$ a structural disturbance vector of
unknown dimension, a nonseparable control variable model takes the
form 
\begin{equation}
Y=g(X,\varepsilon),\label{eq:g(x,e)}
\end{equation}
where $X$ and $\varepsilon$ are independent conditional on an observable
or estimable control variable denoted $V$. Conditioning on the control
variable allows to identify general features of the structural relationship
between $X$ and $Y$ in model (\ref{eq:g(x,e)}), such as those captured
by the structural functions of Blundell and Powell (2003, 2004), and
Imbens and Newey (2009). An important kind of model where $X$ is
independent of $\varepsilon$ conditional on $V$ is a structural
triangular system where $X=h(Z,\eta)$ and $h(z,\eta)$ is one-to-one
in $\eta$. %With the additional
%restriction that 
If $(\varepsilon,\eta)$ are jointly independent of $Z$ then %the triangular system defined by (\ref{eq:g(x,e)}) and (\ref{eq:h(z,eta)})
%implies that 
%$\varepsilon$ and $X$ are then independent conditional on
%the control variable 
$V=F_{X\mid Z}(X\mid Z)$, the conditional CDF of $X$ given $Z$,
is a control variable in this model (Imbens and Newey, 2009).

Leading examples of structural functions are the average structural
function, $\mu(x)$, the distribution structural function, $G(y,x)$,
and the quantile structural function (QSF) $Q(p,x)$, given by 
\begin{eqnarray*}
\mu(x) & := & \int g(x,\varepsilon)F_{\varepsilon}(d\varepsilon),\quad G(y,x):=\text{Pr}(g(x,\varepsilon)\leq y),\\
Q(p,x) & := & p^{\text{th}}\text{ quantile of }g(x,\varepsilon),
\end{eqnarray*}
where $x$ is fixed in these expressions. These structural functions
may be identifiable from control regressions of $Y$ on $X$ and $V$,
including the conditional mean $E[Y\mid X,V],$ CDF, $F_{Y\mid XV}(Y\mid X,V)$,
and %conditional 
quantile function, $Q_{Y\mid XV}(U\mid X,V)$, of $Y$ given $(X,V)$.
In particular, when the support $\mathcal{V}_{x}$ of $V$ conditional
on $X=x$ equals the marginal support $\mathcal{V}$ of $V$ we have
\begin{align}
\mu(x) & =\int_{\mathcal{V}}E[Y\mid X=x,V=v]F_{V}(dv),\quad G(y,x)=\int_{\mathcal{V}}F_{Y\mid XV}(y\mid x,v)F_{V}(dv),\nonumber \\
Q(p,x) & =G^{\leftarrow}(p,x):=\inf\{y\in\mathbb{R}:G(y,x)\geq p\};\label{eq:structure}
\end{align}
see Blundell and Powell (2003) and Imbens and Newey (2009).

The key condition for equation (\ref{eq:structure}) is full support,
that the support $\mathcal{V}_{x}$ of $V$ conditional on $X=x$
equals the marginal support of $V$. Without full support the integrals
would not be well defined because integration would be over a range
of $(x,v)$ values that are outside the joint support of $(X,V).$
Having a full support for each $x$ is equivalent to $(X,V)$ having
rectangular support. In the absence of a rectangular support, global
identification of the structural functions at all $x$ must rely on
alternative conditions that identify $F_{Y\mid XV}(y\mid x,v)$ for
all $(x,v)\in\mathcal{X}\times\mathcal{V}$ and not merely over the
joint support $\mathcal{X}\mathcal{V}$ of $(X,V)$. An example of
such conditions are functional form restrictions on the controlled
regressions $F_{Y\mid XV}$ and $Q_{Y\mid XV}$ which thus constitute
natural modelling targets in the context of nonseparable conditional
independence models. Imbens and Newey (2009) did show that structural
effects may be partially identified without the full support condition.
Here we focus on achieving identification via restricting the form
of control regressions.

We begin with parametric specifications that are linear combinations
of a vector of known functions $w(X,V)$ having the kronecker product
form $p(X)\otimes q(V),$ where $p(X)$ and $q(V)$ are vectors of
transformations of $X$ and $V$, respectively. Let $\Gamma$ denote
a strictly increasing continuous CDF, such as the Gaussian CDF $\Phi$,
with inverse function denoted $\Gamma^{-1}.$ The control regression
specifications we consider are 
\begin{equation}
E[Y|X,V]=\beta_{0}^{\prime}[p(X)\otimes q(V)],\;\;F_{Y|XV}(y|X,V)=\Gamma(\beta(y)^{\prime}[p(X)\otimes q(V)]),\label{eq:param model1}
\end{equation}
and, when $Y$ is continuous, 
\begin{equation}
Q_{Y|X,V}(u|X,V)=\beta(u)^{\prime}[p(X)\otimes q(V)],\quad u\in(0,1),\label{eq:param model2}
\end{equation}
where the coefficients $\beta(y)$ and $\beta(u)$ are functions of
$y$ and $u$, respectively. The quantile and conditional mean coefficients
are related by $\beta_{0}=\int_{0}^{1}\beta(u)du.$ When $Y$ is discrete,
the conditional distribution specification can be thought of as a
discrete choice model as in McFadden (1973). Examples of structural
models that give rise to CRFs of the form (\ref{eq:param model1})-(\ref{eq:param model2})
are given below and in Chernozhukov et al. (2017).

It is convenient in what follows to use a common notation for the
conditional mean, distribution, and quantile control regressions.
For $\mathcal{U}=(0,1)$ and an index set $\mathcal{T}=\{0\}$, $\mathcal{Y}$,
or $\mathcal{U}$, we define the collection of functions indexed by
$\tau\in\mathcal{T}$, 
\[
\varphi_{\tau}(x,v)=\begin{cases}
E[Y\mid X=x,V=v] & \textrm{if \ensuremath{\mathcal{T}}=\{0\}}\\
\Gamma^{-1}\left(F_{Y\mid XV}(\tau\mid x,v)\right) & \textrm{if \ensuremath{\mathcal{T}}=\ensuremath{\mathcal{Y}}}\\
Q_{Y\mid XV}(\tau\mid x,v) & \textrm{if \ensuremath{\mathcal{T}}=\ensuremath{\mathcal{U}}}
\end{cases}.
\]
While the coefficients $y\mapsto\beta(y)$ and $u\mapsto\beta(u)$
in (\ref{eq:param model2}) are infinite-dimensional parameters, for
each $\tau$ in $\mathcal{T}$ the three control regression specifications
share the essentially parametric form 
\[
\varphi_{\tau}(X,V)=\beta_{\tau}^{\prime}w(X,V),\quad w(X,V):=p(X)\otimes q(V),
\]
where the coefficient $\beta_{\tau}$ is a finite-dimensional parameter
vector. This interpretation motivates the following definition of
a \textit{parametric} class of conditional independence models.

\begin{assumption} (a) For the model in (\ref{eq:g(x,e)}), there
exists a control variable $V$ such that $X$ and $\varepsilon$ are
independent conditional on $V$. (b) For a specified set $\mathcal{T}=\{0\}$,
$\mathcal{Y}$, or $\mathcal{U}$, and each $\tau\in\mathcal{T}$,
the outcome $Y$ conditional on $(X,V)$ follows the model 
\begin{equation}
\varphi_{\tau}(X,V)=\beta_{\tau}^{\prime}w(X,V),\quad w(X,V):=p(X)\otimes q(V).\label{eq:PP Second Stage}
\end{equation}
\label{ass:PPTModel} \end{assumption}

Standard results such as those of Newey and McFadden (1994) imply
that point identification of $\beta_{\tau}$ only requires positive
definiteness of the second moment matrix $E[w(X,V)w(X,V)']$. Under
this condition knowledge of the control regressions is achievable
at all $(y,x,v)\in\mathcal{Y}\times\mathcal{X}\times\mathcal{V}$,
and the structural functions are then point identified as functionals
of $\varphi_{\tau}(X,V)$ without full support. The formulation of
primitive conditions under which $E[w(X,V)w(X,V)']$ is positive definite
thus provides a characterisation of the identifying power of parametric
conditional independence models without the full support condition.
Chernozhukov et al. (2017) gave simple sufficient conditions when
the joint distribution of $X$ and $V$ has a continuous component.
Here we generalize these results in a way that allows for the distribution
of $V$ given $X$ (or $X$ given $V)$ to be discrete.

We next give primitive conditions for identification in parametric
conditional independence models. For triangular systems, we show that
these conditions can be satisfied with discrete valued instrumental
variables. Estimation and inference methods for the CRFs in (\ref{eq:PP Second Stage})
and the corresponding structural functions in triangular systems are
extensively analysed by Chernozhukov et al. (2017), and directly apply
when $V$ is observable. 
\begin{rem}
An additional vector of exogenous covariates $Z_{1}$ can be incorporated
straightforwardly in our models. Let $r(Z_{1})$ be a vector of known
transformations of $Z_{1}$, and define $w(X,Z_{1},V):=p(X)\otimes r(Z_{1})\otimes q(V)$
the augmented vector of regressors. The control regressions then take
the form 
\[
\varphi_{\tau}(X,Z_{1},V)=\beta_{\tau}^{\prime}w(X,Z_{1},V),\quad\tau\in\mathcal{T}.
\]
Our identification analysis is not affected by the presence of additional
covariates and for clarity of exposition we do not include them in
the remaining of the paper. Chernozhukov et al. (2017) provide a detailed
exposition of the models we consider in the presence of exogenous
covariates.\qed 
\end{rem}

\section{Identification}

In this section we formulate conditions for positive definiteness
of $E[w(X,V)w(X,V)']$. We first consider the important particular
case where one of the elements $q(V)$ or $p(X)$ of the vector of
regressors $w(X,V)$ is restricted to its first two components. With
either $q(V)=(1,V)'$ or $p(X)=(1,X)'$, each type of restriction
defines a class of baseline parametric models. For triangular systems
%of the form
%(\ref{eq:g(x,e)}) and (\ref{eq:h(z,eta)}) 
we show that a binary instrumental variable is sufficient for identification
of the corresponding control regression and structural functions.
These baseline specifications are thus of substantial interest for
empirical practice, and can be generalised by expanding the restricted
element in $w(X,V)$.

\subsection{Baseline Models\label{subsec:BaselineID}}

In the first class of baseline models, we set $q(V)=(1,V)'$, and
the corresponding vector of regressors in the CRF $\varphi_{\tau}(X,V)$
is $w(X,V)=(p(X)',Vp(X)')'$. We denote the cardinality of sets such
as $\mathcal{X}$ and $\mathcal{V}_{x}$ by $\left|\mathcal{X}\right|$
and $\left|\mathcal{V}_{x}\right|$, respectively. The condition for
identification can then be formulated in terms of the support of $V$
conditional on $X$: letting 
\[
\mathcal{X}_{V}^{o}=\left\{ x\in\mathcal{X}\,:\,|\mathcal{V}_{x}|\geq2\right\} ,
\]
a sufficient condition is that $E[1(X\in\widetilde{\mathcal{X}})p(X)p(X)']$
be positive definite with $\widetilde{\mathcal{X}}\subseteq\mathcal{X}_{V}^{o}$.
Under this condition $\mathcal{X}_{V}^{o}$ is a set with positive
probability, and $V$ has positive variance conditional on $X=x$
for each $x$ in that set.

Alternatively, with $p(X)=(1,X)'$, the vector of regressors in the
CRF $\varphi_{\tau}(X,V)$ that defines the second class of baseline
models is $w(X,V)=(q(V)',Xq(V)')'$. The condition for identification
can then be formulated in terms of the support of $X$ conditional
on $V$: letting 
\[
\mathcal{V}_{X}^{o}=\left\{ v\in\mathcal{V}\,:\,|\mathcal{X}_{v}|\geq2\right\} ,
\]
a sufficient condition is that $E[1(V\in\widetilde{\mathcal{V}})q(V)q(V)']$
be positive definite with $\widetilde{\mathcal{V}}\subseteq\mathcal{V}_{X}^{o}$.
Under this condition $\mathcal{V}_{X}^{o}$ is a set with positive
probability and $X$ has positive variance conditional on $V=v$ for
each $v$ in that set.

Let $C<\infty$ denote some generic positive constant whose value
may vary from place to place.

\begin{assumption} (a) We have that $E[p(X)p(X)']$ exists, $\sup_{x\in\mathcal{X}}E[||q(V)||^{2}\mid X=x]\leq C$
and, for some specified set $\widetilde{\mathcal{X}}$, $E[1(X\in\widetilde{\mathcal{X}})p(X)p(X)']$
is positive definite. (b) We have that $E[q(V)q(V)']$ exists, $\sup_{v\in\mathcal{V}}E[||p(X)||^{2}\mid V=v]\leq C$,
and, for some specified set $\widetilde{\mathcal{V}}$, $E[1(V\in\widetilde{\mathcal{V}})q(V)q(V)']$
is positive definite.\label{ass:P(X)posddef} \end{assumption}

\begin{sloppy}The following theorem states our first main result.
The proofs of all our formal results are given in Appendix \ref{sec:Proofs}. 
\begin{thm}
\label{thm:Theorem1}(i) Let $q(V)=(1,V)'$. If Assumption \ref{ass:P(X)posddef}(a)
holds with $\widetilde{\mathcal{X}}\subseteq\mathcal{X}_{V}^{o}$,
then $E[w(X,V)w(X,V)']$ exists and is positive definite. (ii) Let
$p(X)=(1,X)'$. If Assumption \ref{ass:P(X)posddef}(b) holds with
$\widetilde{\mathcal{V}}\subseteq\mathcal{V}_{X}^{o}$, then $E[w(V,X)w(V,X)']$
exists and is positive definite. 
\end{thm}
The formulation of sufficient conditions for identification in terms
of $\mathcal{X}_{V}^{o}$ and $\mathcal{V}_{X}^{o}$ \foreignlanguage{british}{emphasises}
the fact that the full support condition $\mathcal{V}_{x}=\mathcal{V}$
is not required for $E[w(V,X)w(V,X)']$ to be positive definite in
the baseline specifications, and hence for identification of the control
regressions and structural functions. We also note that identification
does not depend on the dimension of the unrestricted element $p(X)$
or $q(V)$ entering the vector of regressors $w(X,V)$. Thus the baseline
specifications allow for flexible modelling of either how $X$ affects
the CRFs or how $V$ affects the CRFs. When $q(V)=(1,V)'$, complex
features of the relationship between $X$ and $Y$ can also be incorporated
into the specification of the structural functions.

\end{sloppy}

In triangular systems with control variable $V=F_{X\mid Z}(X\mid Z)$,
the conditions given above for $E[w(X,V)w(X,V)']$ to be positive
definite translate into primitive conditions in terms of $\mathcal{Z}_{x}$,
the support of $Z$ conditional on $X=x$. %\footnote{The same holds for any control variable constructed as a strictly
%monotone transformation of $h^{-1}(Z,X)$ in (\ref{eq:h(z,eta)}).} 
Letting 
\[
\mathcal{X}_{Z}^{o}=\left\{ x\in\mathcal{X}\,:\,|\mathcal{Z}_{x}|\geq2\right\} ,
\]
the matrix $E[w(X,V)w(X,V)']$ will be positive definite if Assumption
\ref{ass:P(X)posddef}(a) holds for a set $\widetilde{\mathcal{X}}\subseteq\mathcal{X}_{Z}^{o}$
such that $F_{X | Z}(x | z)\neq F_{X | Z}(x | \tilde{z})$
for some $z,\tilde{z}\in\mathcal{Z}_{x}$ and all $x\in\mathcal{X}_{Z}^{o}$.
For $v\mapsto Q_{X | Z}(v | Z)$ denoting the quantile function
of $X$ conditional on $Z$, the result also holds if Assumption \ref{ass:P(X)posddef}(b)
is satisfied for a set $\widetilde{\mathcal{V}}\subseteq(0,1)$ with
positive probability such that $Q_{X | Z}(v | z)\neq Q_{X | Z}(v | \tilde{z})$
for some $z,\tilde{z}\in\mathcal{Z}$ and all $v\in\widetilde{\mathcal{V}}$.
Under these conditions a discrete instrument, including binary, is
then sufficient for our baseline models to identify the structural
functions. This demonstrates the relevance of the %baseline
baseline specifications in a wide range of empirical settings, for
instance triangular systems with a binary or discrete instrument and
including a discrete or mixed continuous-discrete outcome.\footnote{For example, our baseline models can be used for the specification
of parametric sample selection models with censored selection rule
as considered in Fernandez-Val et al. (2018).}

\subsubsection{Examples\label{subsubsec:Examples}}

An example of a structural model that gives rise to CRFs as in (\ref{eq:PP Second Stage})
is the multidimensional heterogeneous coefficients model\textbf{ 
\begin{equation}
Y=g(X,\varepsilon)=\sum_{j=1}^{J}p_{j}(X)\varepsilon_{j},\;\;E[\varepsilon_{j}|X,V]=E[\varepsilon_{j}|V]=\beta_{0j}'q(V),\;\;j\in\{1,\ldots,J\}.\label{eq:meanindepmodel}
\end{equation}
}The corresponding control mean regression function is
\[
E[Y|X,V]=\sum_{j=1}^{J}p_{j}(X)E[\varepsilon_{j}|X,V]=\sum_{j=1}^{J}p_{j}(X)\{\beta_{0j}'q(V)\}=\beta_{0}'[p(X)\otimes q(V)],
\]
with $\beta_{0}=(\beta_{01}',\ldots,\beta_{0J}')'$, $j\in\{1,\ldots,J\}$,
which has the form of (\ref{eq:PP Second Stage}) with $\mathcal{T}=\{0\}$
and $\tau=0$ in Assumption \ref{ass:PPTModel}. With $q(V)=(1,\widetilde{q}(V)')'$,
where $\widetilde{q}(V)$ is a vector of known functions of $V$ that
satisfy $E[\widetilde{q}(V)]=0$,\footnote{For the baseline specification $q(V)=(1,V)'$, in a triangular model
with $X=h(Z,V)$, $v\mapsto h(Z,v)$ strictly increasing, and $V$
independent from $Z$, the normalisation $V\sim N(0,1)$ implies that
$V=\Phi^{-1}(F_{X|Z}(X|Z))$ is an example of a control variable with
$E[V]=0$. Our identification analysis applies for any strictly monotonic
transformation of the control function $F_{X|Z}(X|Z)$.} the corresponding average structural function takes the form 
\[
\mu(X)=\int_{\mathcal{V}}E[Y\mid X,V=v]F_{V}(dv)=\sum_{j=1}^{J}p_{j}(X)\{\beta_{0j}'E[q(V)]\}=\sum_{j=1}^{J}\beta_{0j1}p_{j}(X),
\]
where $\beta_{0j1}$ denotes the first component of $\beta_{0j}$,
$j\in\{1,\ldots,J\}$. 

When $Y$ is continuous, if the unobserved heterogeneity components
$\varepsilon_{j}$ satisfy the conditional independence property 
\[
\varepsilon_{j}=Q_{\varepsilon_{j}\mid XV}(U\mid X,V)=q(V)'\beta_{j}(U),\quad U\mid X,V\sim U(0,1),\quad j\in\{1,\ldots,J\},
\]
where the unobservable $U$ is the same for each $\varepsilon_{j}$,
%and $(p_{1}(X)\varepsilon_{1},\ldots,p_{J}(X)\varepsilon_{J})$ is
%comonotonic conditional on $(X,V)$,\footnote{A random vector $A=(A_{1},\ldots,A_{J})$ is \textit{comonotonic}
%conditional on a vector $B$ if and only if it can be represented
%as
%\[
%(A_{1},\ldots,A_{J})\mid B=_{d}(Q_{A_{1}\mid B}(U\mid B),\ldots,Q_{A_{J}\mid B}(U\mid B)),\quad U\mid B\sim U(0,1),
%\]
%where $=_{d}$ stands for equality in distribution. For the sum $S(A)\equiv\sum_{j=1}^{J}A_{j}$,
%the vector $A$ then satisfies the quantile additivity property $Q_{S(A)|B}(u|B)=\sum_{j=1}^{J}Q_{A_{j}|B}(u|B)$,
%$u\in(0,1)$.} 
then for each $u\in\mathcal{U}$ the control conditional quantile
function is
\[
Q_{Y\mid XV}(u\mid X,V)=\sum_{j=1}^{J}p_{j}(X)[q(V)'\beta_{j}(u)]=\beta_{u}^{\prime}[p(X)\otimes q(V)],
\]
where $\beta_{u}=(\beta_{1}(u)',\ldots,\beta_{J}(u)')'$, 
which has the form of (\ref{eq:PP Second Stage}) with $\mathcal{T}=\mathcal{U}$
and $\tau=u$ in Assumption \ref{ass:PPTModel}.

Model (\ref{eq:meanindepmodel}) thus allows for flexible modelling
of the relationship between the treatment $X$ and the outcome $Y$
in both the control regression and average structural functions, which
are identified under the conditions of Theorem \ref{thm:Theorem1}.
Similarly, complex features of the relationship between the source
of endogeneity $V$ and the outcome $Y$ can be captured by the model
specification.

An important particular case of model (\ref{eq:meanindepmodel}) with
$p(X)=(1,X)^{\prime}$ is a parametric treatment effects model, where
$p(X)$ is a vector that includes a constant and dummy variables for
various kinds of treatments. A restricted form of the Rosenbaum and
Rubin (1983) treatment effects model is included as a special case,
where $X\in\{0,1\}$ is a treatment dummy variable that is equal to
one if treatment occurs and equals zero without treatment. The control
mean regression for model (\ref{eq:meanindepmodel}) is then 
\[
E[Y|X,V]=E[\varepsilon_{1}|X,V]+E[\varepsilon_{2}|X,V]X=\beta_{01}'q(V)+\{\beta_{02}'q(V)\}X=\beta_{0}'[p(X)\otimes q(V)],
\]
with $\beta_{0}=(\beta_{01}',\beta_{02}')'$. For a set $\widetilde{\mathcal{V}}$
such that $E[1(V\in\widetilde{\mathcal{V}})q(V)q(V)']$ is nonsingular,
a sufficient condition for identification is that the conditional
second moment matrix of $(1,X)^{\prime}$ given $V$ is nonsingular
on $\widetilde{\mathcal{V}}$, which is the same as 
\begin{equation}
\text{Var}(X\mid V)=P(V)[1-P(V)]>0,\quad P(V):=\Pr(X=1\mid V),\label{eq:IdBinary}
\end{equation}
on $\widetilde{\mathcal{V}}$. Here we can see that this identification
condition is the same as $0<P(V)<1$ with positive probability, which
is weaker than the standard identification condition in the unrestricted
model.

In the binary treatment model, $\varepsilon=(\varepsilon_{1},\varepsilon_{2})$
is two dimensional with $\varepsilon_{1}$ giving the outcome without
treatment and $\varepsilon_{2}$ being the treatment effect. Here
the control variables in $V$ would be observable variables such that
the coefficients $(\varepsilon_{1},\varepsilon_{2})$ are mean independent
of treatment conditional on $V$.

\subsection{Generalisation\label{sec:Generalisation}}

We generalise the results above by expanding the set of regressors
in the baseline specifications. In the more general case we consider
here, both $p(X)$ and $q(V)$ are vectors of transformations of $X$
and $V$, respectively. In practice these will typically consist of
basis functions with good approximating properties such as splines,
trigonometric or orthogonal polynomials. %(cf. Appendix  D.2 for an illustration to parametric demand analysis with splines)

One general condition for positive definiteness of $E[w(X,V)w(X,V)']$
is the existence of a set of values $x$ of $X$ with positive probability
such that the smallest eigenvalue of $E[q(V)q(V)'\mid X=x]$ is bounded
away from zero. An alternative general condition is the existence
of a set of values $v$ of $V$ with positive probability such that
the smallest eigenvalue of $E[p(X)p(X)'\mid V=v]$ is bounded away
from zero. This characterisation leads to natural sufficient conditions
for $E[w(X,V)w(X,V)']$ to be positive definite when the vectors $p(X)$
and $q(V)$ are unrestricted.

With $B>0$ denoting some generic constant whose value may vary from
place to place, let $\lambda_{\min}(x)$ denote the smallest eigenvalue
of $E[q(V)q(V)'\mid X=x]$, and define 
\[
\mathcal{X}_{V}^{*}=\left\{ x\in\mathcal{X}\,:\,\lambda_{\min}(x)\geq B>0\right\} .
\]
The smallest eigenvalue of $E[q(V)q(V)'\mid X=x]$ is then bounded
away from zero uniformly over $x\in\mathcal{X}_{V}^{*}$, and a sufficient
condition for identification is that Assumption \ref{ass:P(X)posddef}(a)
holds with $\widetilde{\mathcal{X}}\subseteq\mathcal{X}_{V}^{*}$.
Alternatively, let $\lambda_{\min}(v)$ denote the smallest eigenvalue
of $E[p(X)p(X)'\mid V=v]$, and define 
\[
\mathcal{V}_{X}^{*}=\left\{ v\in\mathcal{V}\,:\,\lambda_{\min}(v)\geq B>0\right\} .
\]
The eigenvalues of $E[p(X)p(X)'\mid V=v]$ are then bounded away from
zero uniformly over $v\in\mathcal{V}_{X}^{*}$, and a sufficient condition
for identification is that Assumption \ref{ass:P(X)posddef}(b) holds
with $\widetilde{\mathcal{V}}\subseteq\mathcal{V}_{X}^{*}$. 
\begin{thm}
For some $B>0$, if either Assumption \ref{ass:P(X)posddef}(a) holds
with $\widetilde{\mathcal{X}}\subseteq\mathcal{X}_{V}^{*}$, or Assumption
\ref{ass:P(X)posddef}(b) holds with $\widetilde{\mathcal{V}}\subseteq\mathcal{V}_{X}^{*}$,
then $E[w(X,V)w(X,V)']$ exists and is positive definite\label{thm:Theorem3} 
\end{thm}
\begin{rem}
\label{rem:Remark2}For the baseline specifications, Proposition \ref{prop:Proposition2}
in Appendix \ref{sec:ProofRemark2} shows that the conditions of Theorem
\ref{thm:Theorem1} satisfy those of Theorem \ref{thm:Theorem3}.
In the simple case $q(V)=(1,V)'$, if Assumption \ref{ass:P(X)posddef}(a)
holds with $\widetilde{\mathcal{X}}\subseteq\mathcal{X}_{V}^{o}$
then $\text{Var}(V\mid X=x)\geq B>0$ for each $x\in\mathcal{X}_{V}^{o}$,
and Assumption \ref{ass:P(X)posddef}(a) also holds with $\widetilde{\mathcal{X}}\subseteq\mathcal{X}_{V}^{*}$.
In the simple case $p(X)=(1,X)'$, if Assumption \ref{ass:P(X)posddef}(b)
holds with $\widetilde{\mathcal{V}}\subseteq\mathcal{V}_{X}^{o}$
then $\text{Var}(X\mid V=v)\geq B>0$ for each $v\in\mathcal{V}_{X}^{o}$,
and Assumption \ref{ass:P(X)posddef}(b) also holds with $\widetilde{\mathcal{X}}\subseteq\mathcal{X}_{V}^{*}$. 
\end{rem}

\subsection{Discussion}

Theorem \ref{thm:Theorem3} gives a general identification result
for models with regressors of a kronecker product form $w(X,V)=p(X)\otimes q(V)$.
By standard results such as those of Newey and McFadden (1994), $\beta_{\tau}$ in (\ref{eq:PP Second Stage}) is identified for each $\tau\in\mathcal{T}$, and positive
definiteness of the matrix $E\left[w(V,X)w(V,X)^{\prime}\right]$
is then a sufficient condition for uniqueness of the CRFs with probability
one. Thus the conditions of Theorem \ref{thm:Theorem3} are also sufficient
for the models we consider to identify their corresponding structural
functions. 
\begin{thm}
Suppose the assumptions of Theorem \ref{thm:Theorem3} are satisfied.
If Assumption \ref{ass:PPTModel} holds with $\mathcal{T}=\mathcal{Y}$
or $\mathcal{U}$ then the average, distribution and quantile structural
functions are identified. If Assumption \ref{ass:PPTModel} holds
with $\mathcal{T}=\{0\}$ then the average structural function is
identified.\label{thm:Theorem4} 
\end{thm}
The formulation of identification conditions in terms of the second
conditional moment matrices of $p(X)$ and $q(V)$ is a considerable
simplification relative to existing conditions in the literature.
The assumptions of Theorems \ref{thm:Theorem1}-\ref{thm:Theorem4}
are more primitive and easier to interpret than the dominance condition
proposed by Chernozhukov et al. (2017) for positive definiteness of
$E[w(X,V)w(X,V)']$.\footnote{Chernozhukov et al. (2017) assume that the joint probability distribution
of $X$ and $V$ dominates a product probability measure $\mu(x)\times\rho(v)$
such that $E_{\mu}[p(X)p(X)']$ and $E_{\rho}[q(V)q(V)']$ are positive
definite. This condition is sufficient for $E[w(X,V)w(X,V)']$ to
be positive definite, but is difficult to interpret.} For instance, for the baseline specifications these assumptions provide
transparent testable implications using empirical estimates of common
statistical objects, for both triangular systems (e.g., $Q_{X|Z}(v|z)$
and $F_{X|Z}(x|z)$ in Section \ref{subsec:BaselineID}) and treatment
effect models (e.g., $P(V)$ in condition (\ref{eq:IdBinary})). These
conditions are also weaker than the %rectangular 
full support condition %assumption of Imbens and Newey (2009),
or the measurable separability condition of Florens et al. (2008),
which require the control variable to have a continuous distribution
%support 
conditional on $X$.

In a triangular system with control variable $V=F_{X\mid Z}(X\mid Z)$,
our identification conditions %can be satisfied by a discrete valued
%instrument, and 
admit an equivalent formulation in terms of the first
stage model and the instrument $Z$. Letting $\widetilde{\lambda}_{\min}(x)$
denote the smallest eigenvalue of 
\[
E[q(F_{X\mid Z}(X\mid Z))q(F_{X\mid Z}(X\mid Z))^{\prime}\mid X=x],
\]
for $x\in\mathcal{X}$, for some $B>0$ define the corresponding set $\mathcal{X}_{Z}^{*}=\{x\in\mathcal{X}\,:\,\widetilde{\lambda}_{\min}(x)\geq B>0\}$.
Then $\widetilde{\lambda}_{\min}(x)=\lambda_{\min}(x)$ and $\mathcal{X}_{Z}^{*}=\mathcal{X}_{V}^{*}$.
Thus Assumption \ref{ass:P(X)posddef}(a) with $\widetilde{\mathcal{X}}\subseteq\mathcal{X}_{Z}^{*}$
is sufficient for identification by Theorem \ref{thm:Theorem3}. Alternatively,
letting $\widetilde{\lambda}_{\min}(v)$ denote the smallest eigenvalue
of 
\[
E[p(Q_{X\mid Z}(v\mid Z))p(Q_{X\mid Z}(v\mid Z))^{\prime}],
\]
for $v\in(0,1)$, for some $B>0$ define the corresponding set $\mathcal{V}_{Z}^{*}=\{v\in(0,1)\,:\,\widetilde{\lambda}_{\min}(v)\geq B>0\}$.
Then, by independence of $V$ from $Z$, $\widetilde{\lambda}_{\min}(v)=\lambda_{\min}(v)$
and $\mathcal{V}_{Z}^{*}=\mathcal{V}_{X}^{*}$. Thus Assumption \ref{ass:P(X)posddef}(b)
with $\widetilde{\mathcal{V}}\subseteq\mathcal{V}_{Z}^{*}$ is sufficient
for identification by Theorem \ref{thm:Theorem3}. %\begin{prop}
%Suppose that $|\mathcal{Z}|<\infty$. Then Assumptions \ref{ass:P(X)posddef}(a)
%and \ref{ass:P(X)posddef}(b) hold with $\widetilde{\mathcal{X}}\subseteq\mathcal{X}_{Z}^{*}$
%and $\widetilde{\mathcal{V}}\subseteq\mathcal{V}_{Z}^{*}$, respectively,
%only if there exist sets $\widetilde{\mathcal{X}}\subseteq\mathcal{X}$
%and $\widetilde{\mathcal{V}}\subseteq(0,1)$ of positive probability
%such that $\inf_{x\in\widetilde{\mathcal{X}}}|\mathcal{R}(x)|\geq K$
%and $\inf_{v\in\widetilde{\mathcal{V}}}|\mathcal{Q}(v)|\geq J$, respectively.\label{prop:Proposition1} 
%\end{prop}
%Proposition \ref{prop:Proposition1} shows that the flexibility of
%%parametric 
%triangular systems is restricted by the cardinality of
%the set of instrumental values. Thus identification can be achieved
%in the presence of a binary instrument only when one of the two vectors
%$q(V)$ and $p(X)$ is of dimension two. More generally, identification
%in the class of models we consider cannot be achieved whenever $|\mathcal{Z}|<\min(J,K)$.

\section{Partially Nonparametric Specifications}

An important generalisation of the parametric specifications of the
previous section is one where either the relationship between $X$
and $Y$ or between $V$ and $Y$ is unspecified in the CRFs. This
gives rise to two classes of models with known functional form of
either how $X$ affects the CRFs or how $V$ affects the CRFs, but
not both. These models are special cases of functional coefficient
regression models.

The first class of partially nonparametric models we consider is one
where $X$ is known to affect the CRF $\varphi_{\tau}(X,V)$ only
through a vector of known functions $p(X)$. We assume that 
\begin{equation}
\varphi_{\tau}(X,V)=p(X)^{\prime}q_{\tau}(V),\quad\tau\in\mathcal{T},\label{eq:Model1}
\end{equation}
where the vector of functions $q_{\tau}(V)$ is now unknown, rather
than a linear combination of finitely many known transformations of
$V$. An example of a structural model that gives rise to CRFs as
in (\ref{eq:Model1}) is the heterogeneous coefficients model 
\[
Y=g(X,\varepsilon)=p(X)'\varepsilon,\quad E[\varepsilon\mid X,V]=E[\varepsilon\mid V],\quad E[\varepsilon\mid V]=:q_{0}(V).
\]
This model is studied in Masten and Torgovitsky (2016)\textbf{ }and
Newey and Stouli (2018), and generalises the polynomial specifications
of Florens et al. (2008) to allow $p(X)$ to be any functions of $X$
rather than just powers of $X$. The corresponding mean CRF of $Y$
conditional on $(X,V)$ is 
\begin{equation}
E\left[Y\mid X,V\right]=p(X)^{\prime}E\left[\varepsilon\mid X,V\right]=p(X)^{\prime}E\left[\varepsilon\mid V\right]=p(X)^{\prime}q_{0}(V),\label{eq:eps_j-1-1}
\end{equation}
which has the form of (\ref{eq:Model1}) with $\mathcal{T}=\{0\}$
and $\tau=0$. When the outcome $Y=\sum_{j=1}^{J}p_{j}(X)\varepsilon_{j}$
is continuous, if the unobserved heterogeneity components $\varepsilon_{j}$
further satisfy %\textbf{conditional mean independence property [THIS CASE IS IN THE OTHER PAPER - SHALL WE REMOVE IT?]}
%\begin{equation}
%E[\varepsilon _{j}|X,V]=E[\varepsilon _{j}|V],\;j\in \{1,\ldots ,J\};
%\label{eq:eps_j-1}
%\end{equation}%
%%see Newey and Stouli (2018) for detailed treatment of this case. %This model is a partially nonparametric generalisation of (\ref{eq:eps_j}). 
%Alternatively, under 
the %stronger 
conditional independence property 
\begin{equation}
\varepsilon_{j}=Q_{\varepsilon_{j}\mid XV}(U\mid X,V)=Q_{\varepsilon_{j}\mid V}(U\mid V),\;U\mid X,V\sim U(0,1),\;j\in\{1,\ldots,J\},\label{eq:eps_j-1}
\end{equation}
where the unobservable $U$ is the same for each $\varepsilon_{j}$,
%with $(p_{1}(X)\varepsilon_{1},\ldots,p_{J}(X)\varepsilon_{J})$ comonotonic
%conditional on $(X,V)$, 
then the control quantile regression function of $Y$ conditional
on $(X,V)$ is 
\[
Q_{Y\mid XV}(u\mid X,V)=\sum_{j=1}^{J}p_{j}(X)Q_{\varepsilon_{j}\mid V}(u\mid V)=p(X)^{\prime}q_{u}(V),\quad u\in\mathcal{U},
\]
with $q_{u}(v):=(Q_{\varepsilon_{1}\mid V}(u\mid v),\ldots Q_{\varepsilon_{J}\mid V}(u\mid v))'$,
which has the form of (\ref{eq:Model1}) with $\mathcal{T}=\mathcal{U}$
and $\tau=u$. Thus this is a model with known functional form of
how $X$ affects the control conditional mean and quantile functions.

The second class of partially nonparametric models we consider is
one where $V$ is known to affect the CRF $\varphi_{\tau}(X,V)$ only
through a vector of known functions $q(V)$. We assume that 
\begin{equation}
\varphi_{\tau}(X,V)=p_{\tau}(X)^{\prime}q(V),\quad\tau\in\mathcal{T},\label{eq:Model2}
\end{equation}
where the vector of functions $p_{\tau}(X)$ is now unknown, rather
than just a linear combination of finitely many known transformations
of $X$. An example of a structural model that gives rise to CRFs
as in (\ref{eq:Model2}) is the heterogeneous coefficients model 
\[
Y=g(X,\varepsilon)=p_{0}(X)'\varepsilon,\quad E[\varepsilon\mid X,V]=E[\varepsilon\mid V],\quad E[\varepsilon\mid V]=q(V),
\]
where $p_{0}(X)$ is a vector of unknown functions, while $q(V)$
is a vector of known functions. In the simplest case with $q(V)=(1,V)'$,
the corresponding mean CRF of $Y$ conditional on $(X,V)$ is 
\begin{equation}
E\left[Y\mid X,V\right]=p_{0}(X)'E\left[\varepsilon\mid X,V\right]=p_{0}(X)^{\prime}E\left[\varepsilon\mid V\right]=p_{0}(X)^{\prime}q(V),\label{eq:eps_j-1-1-1}
\end{equation}
which has the form of (\ref{eq:Model2}) with $\mathcal{T}=\{0\}$
and $\tau=0$. 

With $V$ normalised to satisfy $E[V]=0$, the corresponding average
structural function takes the form 
\[
\mu(X)=\int_{\mathcal{V}}\{p_{0}(X)^{\prime}q(v)\}F_{V}(dv)=p_{01}(X)+p_{02}(X)E[V]=p_{01}(X).
\]

Specifications (\ref{eq:eps_j-1-1}) and (\ref{eq:eps_j-1-1-1}) illustrate
the range of models allowed by partially nonparametric specifications.
For treatment effect models, the choice of specification (\ref{eq:eps_j-1-1})
is dictated by the definition of $X$ as a vector of dummy variables
for each treatment, which are known functions of $X$. For triangular
models, the choice of specification (\ref{eq:eps_j-1-1-1}) allows
for a fully flexible average structural function specification, while
restricting the relationship between the CRFs and $V$ to belong to
a known class of functions, e.g., to be linear when $q(V)=(1,V)'$.
In practice, a richer support of the instrument will allow for a more
flexible relationship, and hence make the choice of either class of
CRFs less restrictive. When the instrument takes a small number of
values, existing model selection methods such as $\ell_{1}$-penalized
quantile (Belloni and Chernozhukov, 2011), distribution (Belloni et
al., 2017), and mean regression (Tibshirani, 1996) provide natural
avenues for empirical specification of CRFs. 
\begin{rem}
Additional exogenous covariates $Z_{1}$ can be incorporated straightforwardly
in these models through the known functional component of the CRF
$\varphi_{\tau}(X,V)$. With an exogenous vector of covariates $Z_{1}$,
model (\ref{eq:Model1}) takes the form 
\[
\varphi_{\tau}(X,Z_{1},V)=p(X,Z_{1})^{\prime}q_{\tau}(V),
\]
where $p(X,Z_{1})$ is a vector of known functions of $(X,Z_{1})$,
and model (\ref{eq:Model2}) takes the form 
\[
\varphi_{\tau}(X,Z_{1},V)=p_{\tau}(X)^{\prime}q(Z_{1},V),
\]
where $q(Z_{1},V)$ is a vector of known functions of $(Z_{1},V)$.\qed 
\end{rem}
\begin{sloppy}

The following assumption gathers the two classes of partially nonparametric
specifications.

\begin{assumption}(a) For a specified set $\mathcal{T}=\{0\}$, $\mathcal{Y}$,
or $\mathcal{U}$, and each $\tau\in\mathcal{T}$, the outcome $Y$
conditional on $(X,V)$ follows the model 
\begin{equation}
\varphi_{\tau}(X,V)=p(X)'q_{\tau}(V);\label{eq:Model1_2}
\end{equation}
we have $E\left[Y^{2}\right]<\infty$ and $E[||p(X)||^{2}]<\infty$;
and $E[p(X)p(X)' | V]$ exists and is nonsingular with probability
one; or (b) for a specified set $\mathcal{T}=\{0\}$, $\mathcal{Y}$,
or $\mathcal{U}$, and each $\tau\in\mathcal{T}$, the outcome $Y$
conditional on $(X,V)$ follows the model 
\begin{equation}
\varphi_{\tau}(X,V)=q(V)'p_{\tau}(X);\label{eq:Model2_2}
\end{equation}
we have $E\left[Y^{2}\right]<\infty$ and $E[||q(V)||^{2}]<\infty$;
and $E[q(V)q(V)' | X]$ exists and is nonsingular with probability
one.\label{ass:NPDRModel}\end{assumption}

The next result states our main identification result of this section. 
\begin{thm}
(i) If Assumption \ref{ass:NPDRModel}(a) holds then $q_{\tau}\left(V\right)$
is identified for each $\tau\in\mathcal{T}$. (ii) If Assumption \ref{ass:NPDRModel}(b)
holds then $p_{\tau}\left(X\right)$ is identified for each $\tau\in\mathcal{T}$.\label{thm:Theorem5} 
\end{thm}
We earlier discussed conditions for nonsingularity of $E[p(X)p(X)'\mid V]$
and $E[q(V)q(V)'\mid X]$. All those conditions are sufficient for
identification of $q_{\tau}(V)$ and $p_{\tau}(X)$, including those
that allow for discrete valued instrumental variables, under the important
stricter condition that they hold on sets of $V$ and $X$ having
probability one, respectively. We also note that identification of
$q_{\tau}(V)$ and $p_{\tau}(X)$ means uniqueness on sets of $V$
and $X$ having probability one, respectively. Thus the structural
functions corresponding to models (\ref{eq:Model1_2}) and (\ref{eq:Model2_2})
are identified. For example, in the first class of models the quantile
and distribution structural functions will be identified as 
\[
Q(p,X)=G^{\leftarrow}(p,X),\quad G(y,X)=\int_{\mathcal{V}}\Gamma\left(p(X)'q_{y}(v)\right)F_{V}(dv),
\]
since $p(X)$ and $\Gamma$ are known functions and $q_{y}(V)$ is
identified, and hence $\Gamma(p(X)'q_{y}(V))$ also is. 
\begin{thm}
Suppose Assumption \ref{ass:PPTModel}(a) holds. If Assumption \ref{ass:NPDRModel}
holds with $\mathcal{T}=\mathcal{Y}$ or $\mathcal{U}$ then the average,
distribution and quantile structural functions are identified. If
Assumption \ref{ass:NPDRModel} holds with $\mathcal{T}=\{0\}$ then
the average structural function is identified.\label{thm:Theorem6} 
\end{thm}
\end{sloppy}

\section{Empirical Application\label{sec:Empirical-Application}}

In this section we illustrate our identification results by estimating
the QSF for a %parametric 
triangular system for Engel curves. We focus on the structural relationship
between household's total expenditure and household's demand for two
goods: food and leisure. We take the outcome $Y$ to be the expenditure
share on either food or leisure, and $X$ the logarithm of total expenditure.
We use as an instrument a discretised version $\widetilde{Z}$ of
the logarithm of gross earnings of the head of household $Z^{*}$.
We also include an additional binary covariate $Z_{1}$ accounting
for the presence of children in the household.

There is a large literature using nonseparable triangular systems
for the identification and estimation of Engel curves (Imbens and
Newey, 2009; Chernozhukov et al., 2015, 2017). We follow Chernozhukov
et al. (2017) who consider estimation of structural functions for
food and leisure using triangular control regression specifications
in kronecker product form. For comparison purposes we use the same
dataset from Blundell et al. (2007), the 1995 U.K. Family Expenditure
Survey. We restrict the sample to 1,655 married or cohabiting couples
with two or fewer children, in which the head of the household is
employed and between the ages of 20 and 55 years. For this sample
we estimate the QSF for both goods using discrete instruments, and
then compare our results to those obtained with a continuous instrument
by Chernozhukov et al. (2017).

%\begin{sloppy}
We consider the triangular system, 
\begin{align*}
Y & =Q_{Y\mid XV}(u\mid X,V)=\beta(U)^{\prime}[p(X)\otimes r(Z_{1})\otimes q(V)],\quad U\mid X,Z_{1},V\sim U(0,1)\\
X & =Q_{X\mid Z}(V\mid Z)=\pi(V)^{\prime}[s(\widetilde{Z})\otimes r(Z_{1})],\quad V\mid Z\sim U(0,1),\quad Z:=(\widetilde{Z},Z_{1})^{\prime},
\end{align*}
where $s(\widetilde{Z})=(1,\widetilde{Z})^{\prime}$, $r(Z_{1})=(1,Z_{1})^{\prime}$,
$p(X)=(1,X)^{\prime}$ and $q(V)=(1,\Phi^{-1}(V))^{\prime}$. The
corresponding QSFs are estimated by the quantile regression estimators
of Chernozhukov et al. (2017), described in Appendix \ref{sec:Quantile-Regression-Estimation}.
For our sample of $n=1,655$ observations $\left\{ (Y_{i},X_{i},Z_{i})\right\} _{i=1}^{n}$,
we construct two sets of four discrete valued instruments taking $M=2,3,5$
and $15$ values, respectively, and then estimate the QSFs using one
instrument at a time.\footnote{The design with discretised instruments might not be consistent with
the original specification, which is linear in the continuous instrument
$\widetilde{Z}$. Nonetheless, overall the empirical results appear
to be robust to discretisation of the instrument.} In the first set the instrument $\widetilde{Z}$ is uniformly distributed
across its support (\textsc{Design} 1). For $t_{m}=m/M$, $m\in\{0,1,\ldots,M\}$,
let $\widehat{Q}_{Z^{*}}(t_{m})$ denote the sample $t_{m}$ quantile
of $Z^{*}$. For $i\in\{1,\ldots,n\}$ and $m\in\{0,1,\ldots,M-1\}$
such that $Z_{i}^{*}\in[\widehat{Q}_{Z^{*}}(t_{m}),\widehat{Q}_{Z^{*}}(t_{m+1}))$,
we define 
\[
\widetilde{Z}_{i}=\widehat{Q}_{Z^{*}}(t_{m})+\frac{1}{2}\left[\widehat{Q}_{Z^{*}}(t_{m+1})-\widehat{Q}_{Z^{*}}(t_{m})\right].
\]
For an observation $i$ such that $Z_{i}^{*}=\max_{i\leq n}(Z_{i}^{*})$,
we define %$\widetilde{Z}_{i}=\widehat{Q}_{Z^{*}}(t_{M-1})+\frac{1}{2}\left[\widehat{Q}_{Z^{*}}(t_{M})-\widehat{Q}_{Z^{*}}(t_{M-1})\right]$.
$\widetilde{Z}_{i}=\widehat{Q}_{Z^{*}}(t_{M-1})+\frac{1}{2}[\widehat{Q}_{Z^{*}}(t_{M})-\widehat{Q}_{Z^{*}}(t_{M-1})]$.
In the second set the instrument $\widetilde{Z}$ is discretised according
to a non uniform distribution (\textsc{Design} 2). Define the equispaced
grid $\min_{i\leq n}(Z_{i}^{*})=\xi_{0}<\xi_{1}<\ldots<\xi_{M}=\max_{i\leq n}(Z_{i}^{*})$.
For $i\in\{1,\ldots,n\}$ and $m\in\{0,\ldots,M-1\}$ such that $Z_{i}^{*}\in[\xi_{m},\xi_{m+1})$
we define 
\[
\widetilde{Z}_{i}=\xi_{m}+\frac{1}{2}\left[\xi_{m+1}-\xi_{m}\right].
\]
For an observation $i$ such that $Z_{i}^{*}=\max_{i\leq n}(Z_{i}^{*})$,
we define $\widetilde{Z}_{i}=\xi_{M-1}+\frac{1}{2}\left[\xi_{M}-\xi_{M-1}\right]$.
%\par\end{sloppy}

\begin{figure}[t]
\subfloat[$M=2$.]{\includegraphics[width=7.5cm,height=6cm]{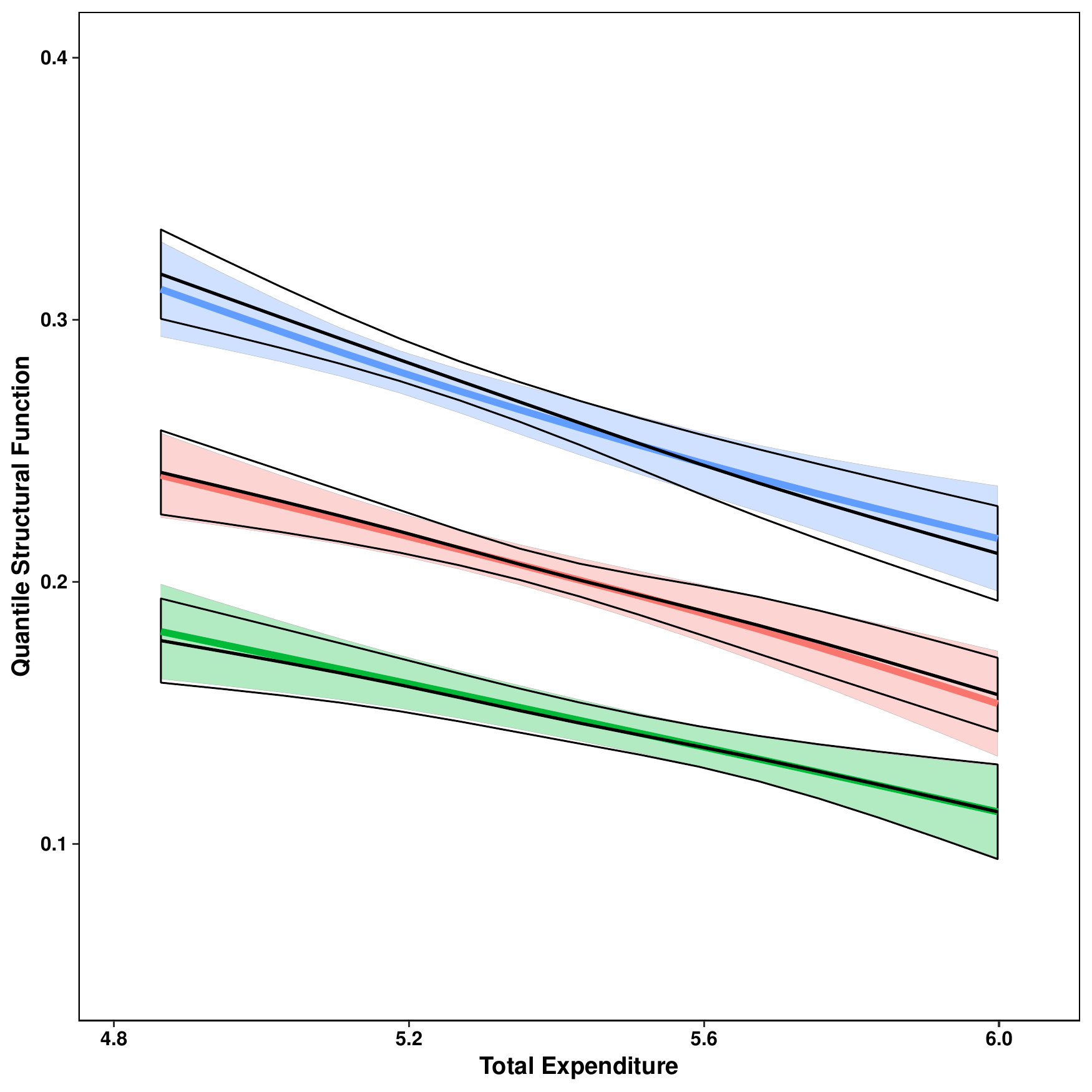} 

}\subfloat[$M=3$.]{\includegraphics[width=7.5cm,height=6cm]{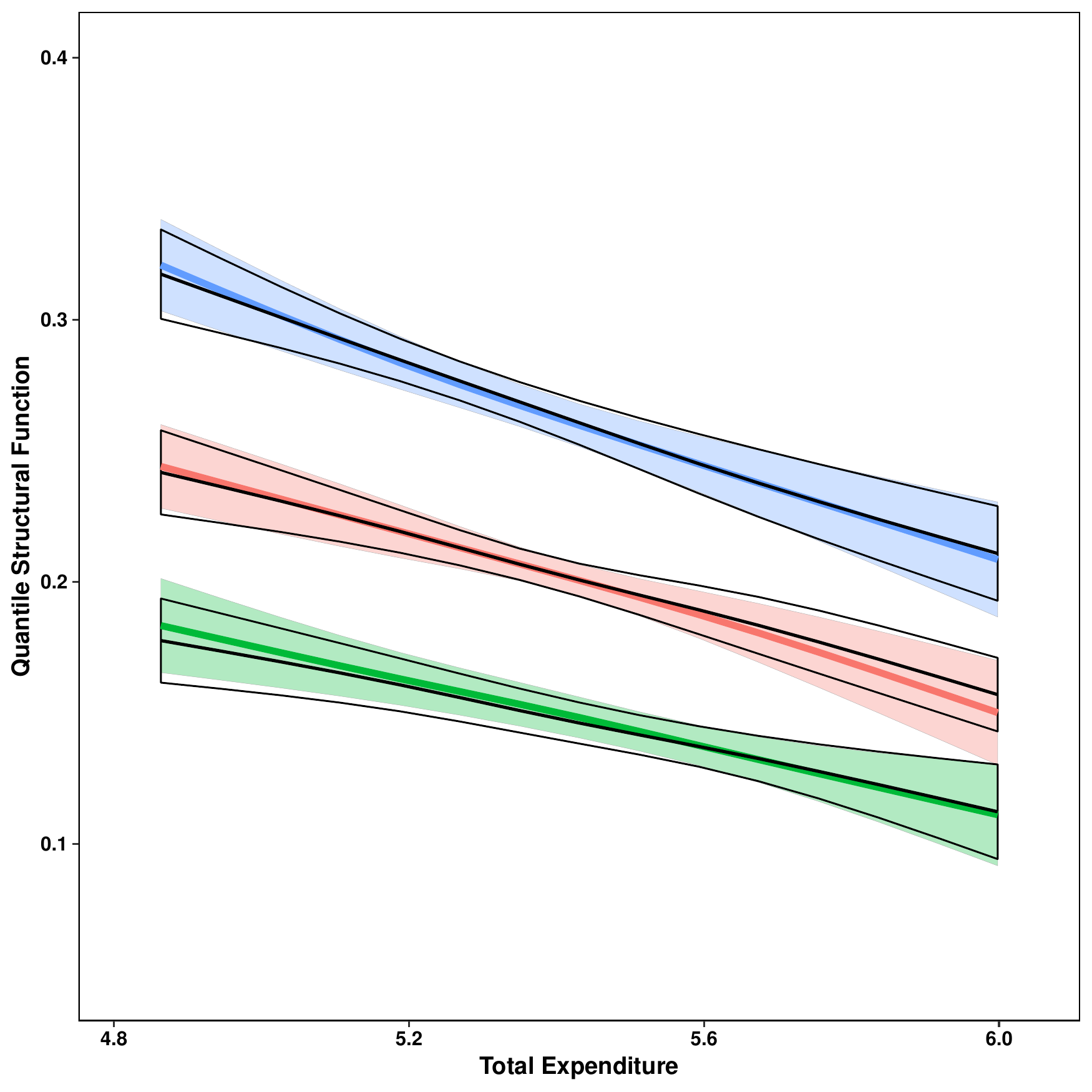} 

}

\subfloat[$M=5$.]{\includegraphics[width=7.5cm,height=6cm]{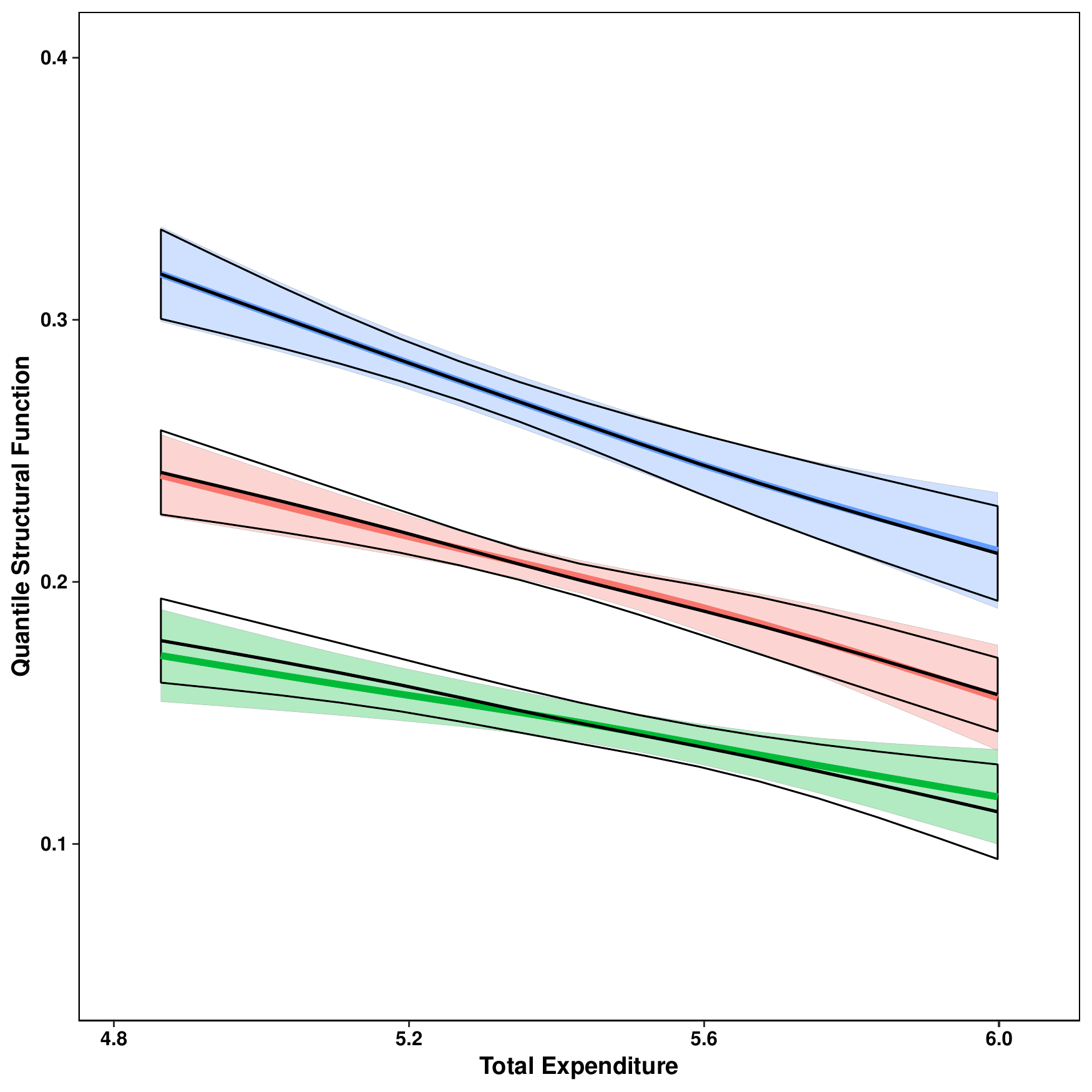} 

}\subfloat[$M=15$.]{\includegraphics[width=7.5cm,height=6cm]{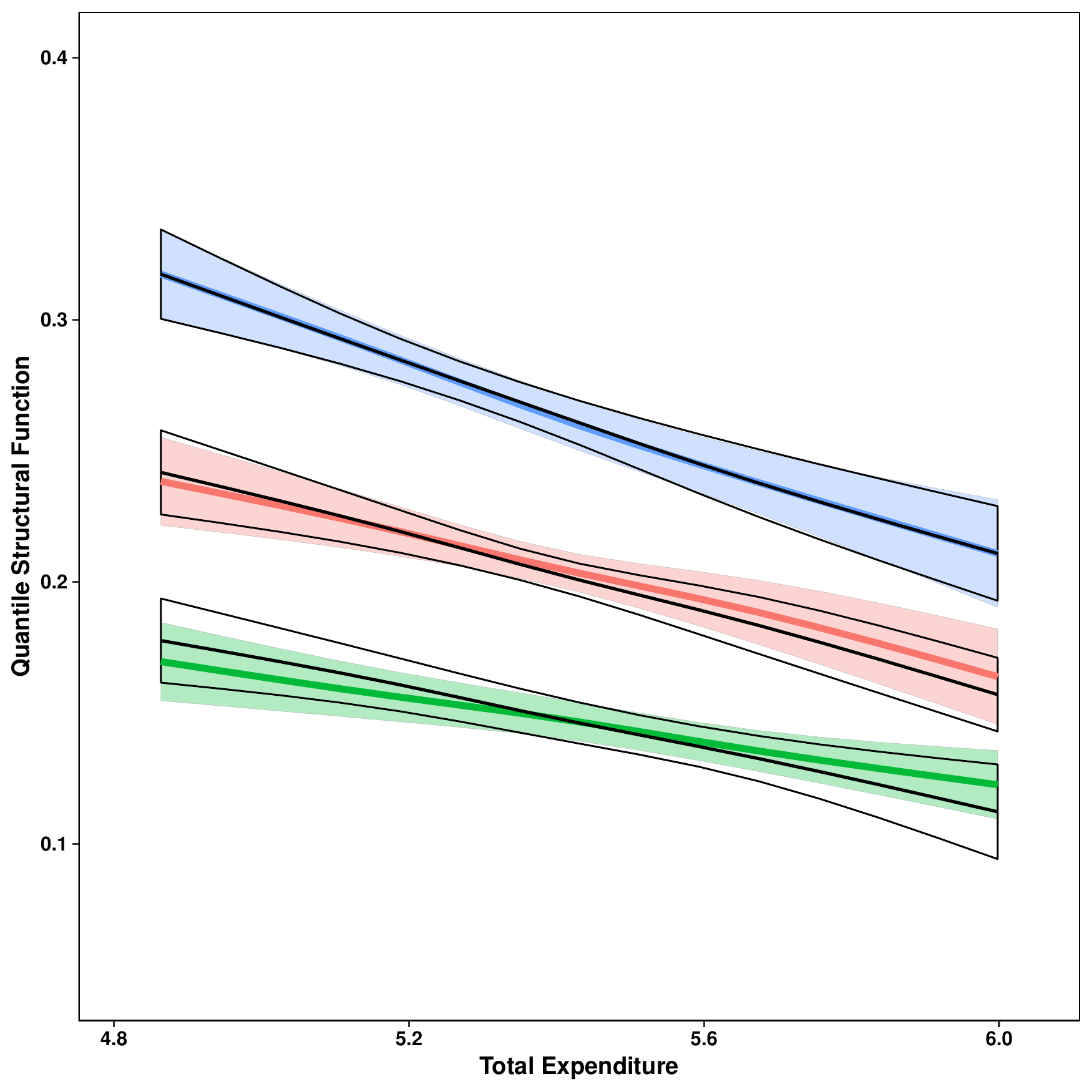} 

}

\caption{\textsc{Design} 1. QSF for food with discrete instrument $\widetilde{Z}$
(\foreignlanguage{british}{coloured}) and with continuous instrument
$Z^{*}$ (black).}
\label{fig:QSF_Des1_food} 
\end{figure}

\begin{figure}[t]
\subfloat[$M=2$.]{\includegraphics[width=7.5cm,height=6cm]{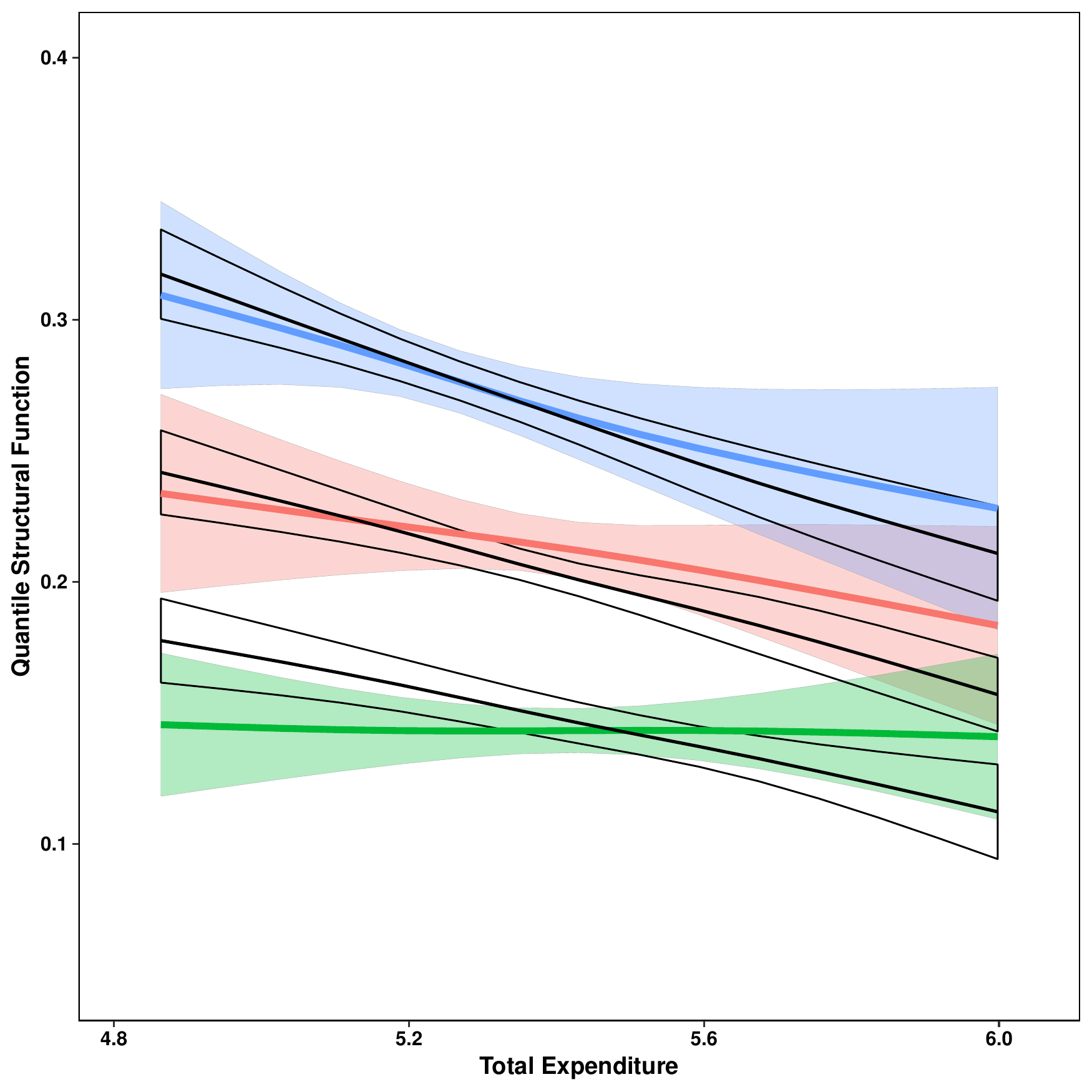} 

}\subfloat[$M=3$.]{\includegraphics[width=7.5cm,height=6cm]{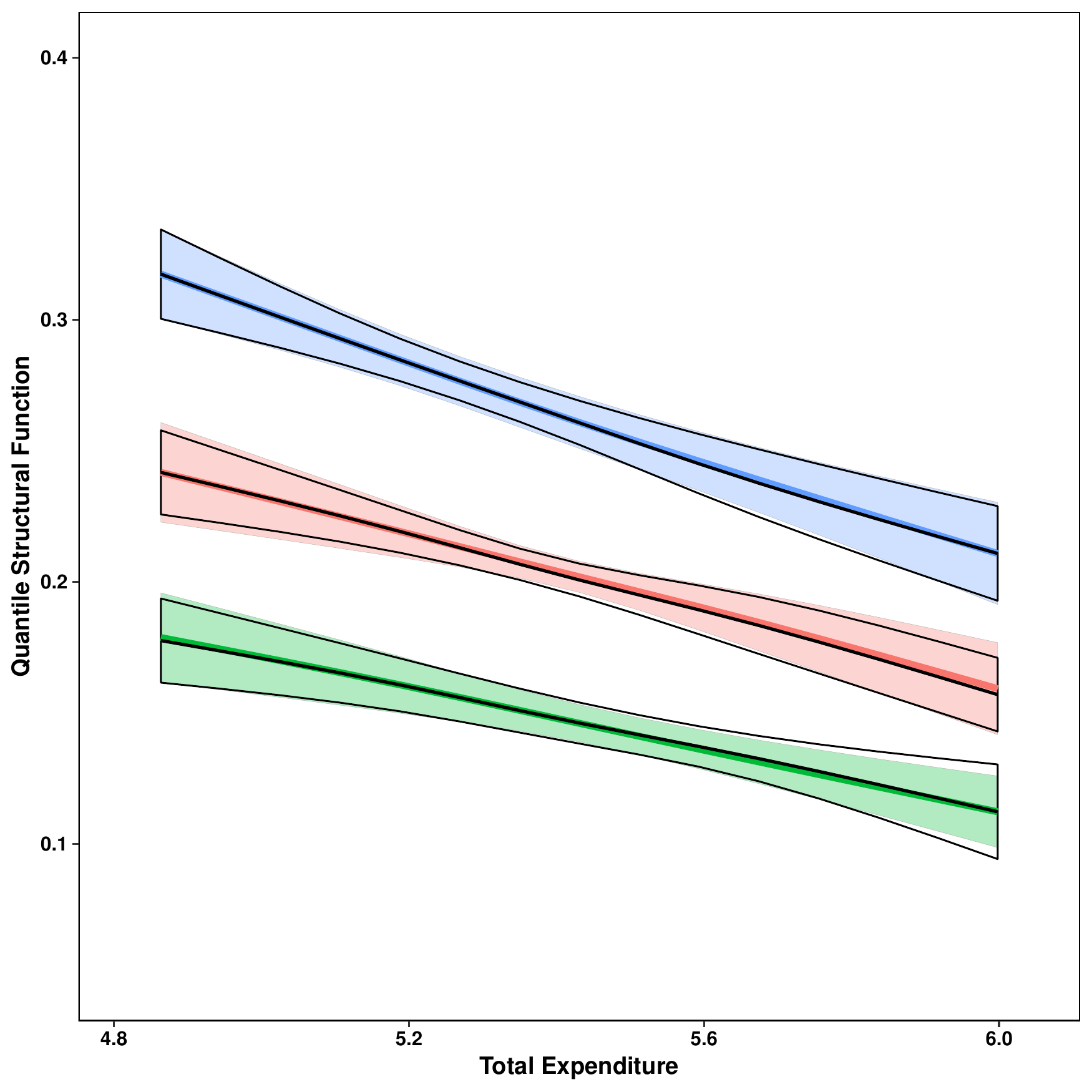}

}

\subfloat[$M=5$.]{\includegraphics[width=7.5cm,height=6cm]{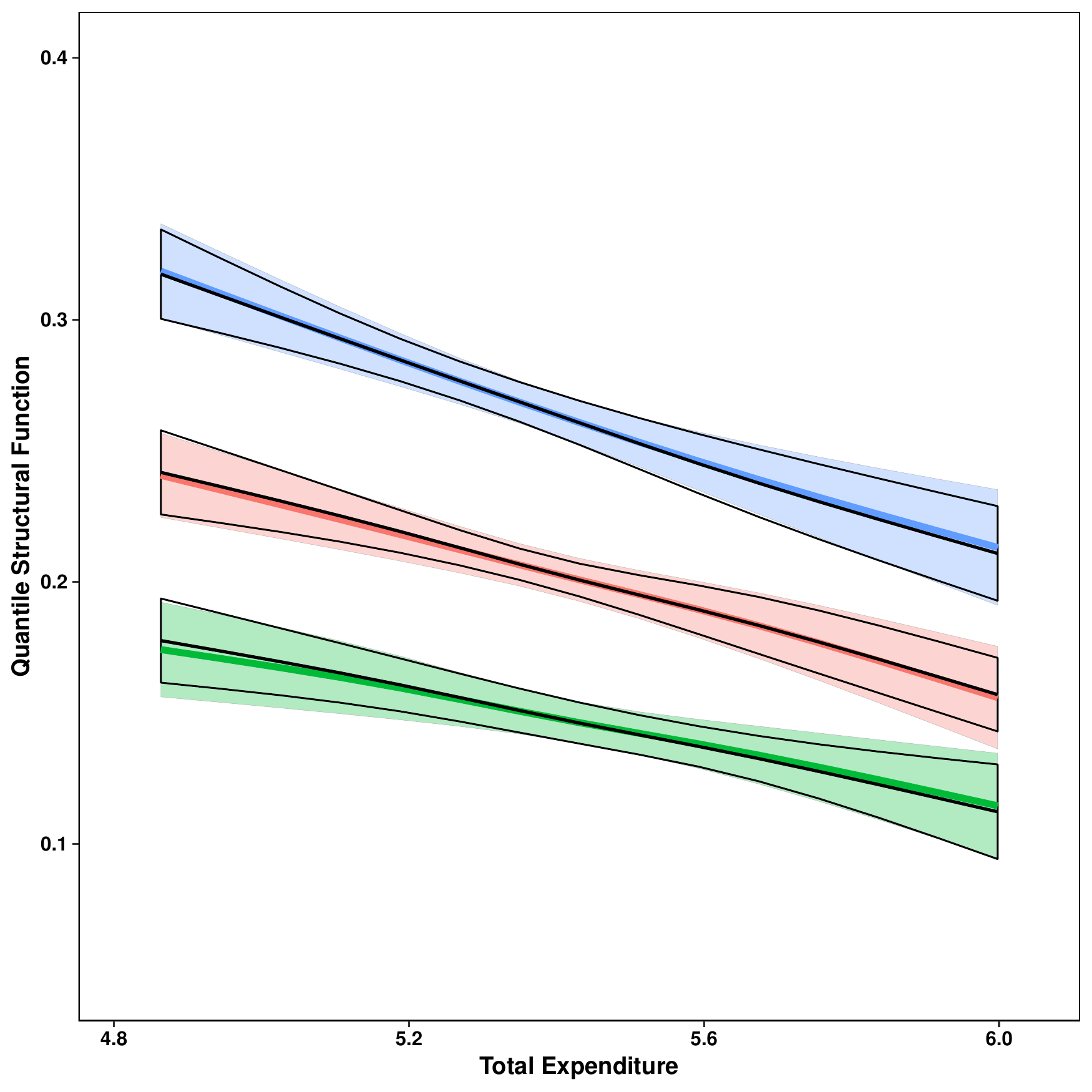} 

}\subfloat[$M=15$.]{\includegraphics[width=7.5cm,height=6cm]{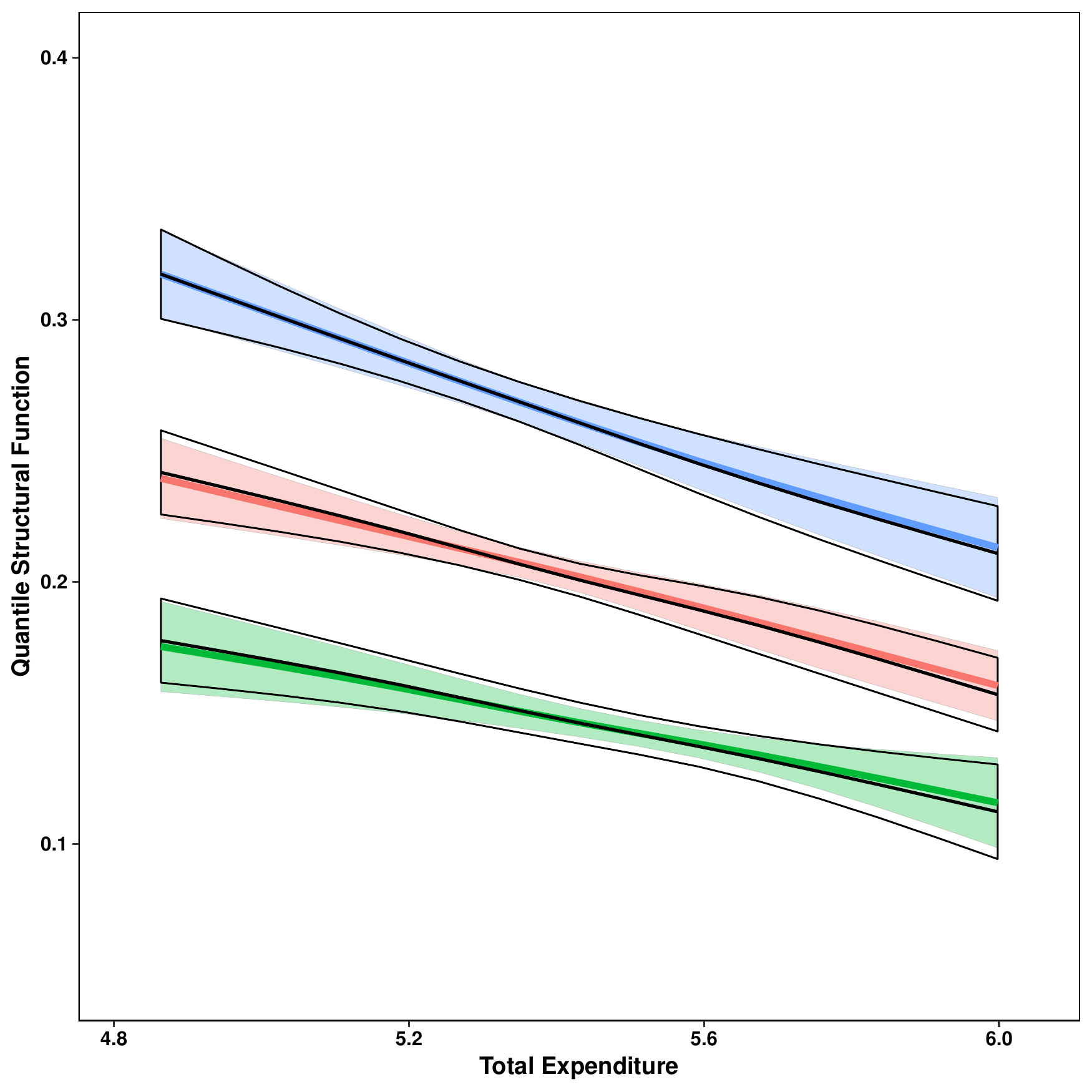} 

}

\caption{\textsc{Design} 2. QSF for food with discrete instrument $\widetilde{Z}$
(\foreignlanguage{british}{coloured}) and with continuous instrument
$Z^{*}$ (black).}
\label{fig:QSF_Des2_food} 
\end{figure}

\begin{figure}[t]
\subfloat[$M=2$.]{\includegraphics[width=7.5cm,height=6cm]{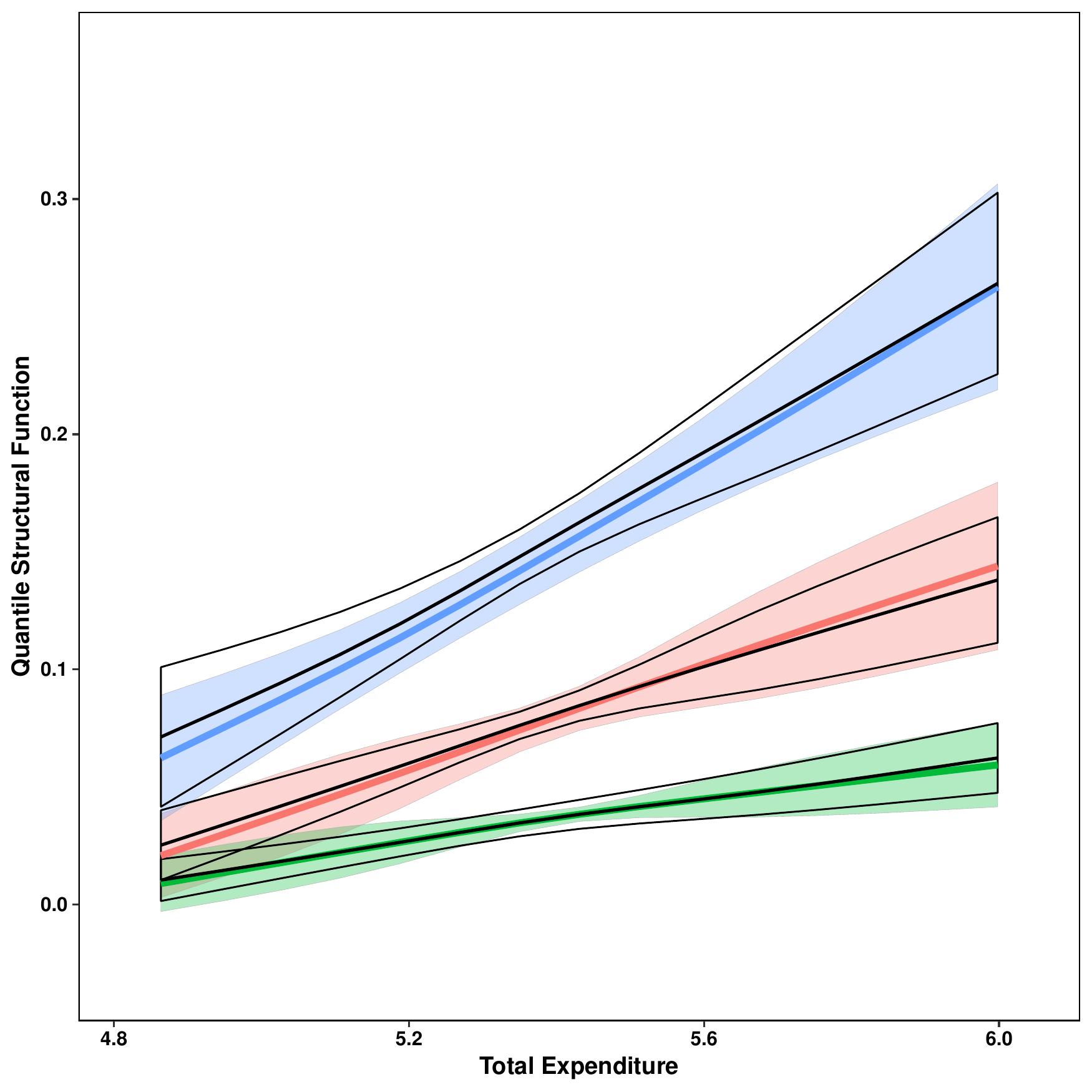} 

}\subfloat[$M=3$.]{\includegraphics[width=7.5cm,height=6cm]{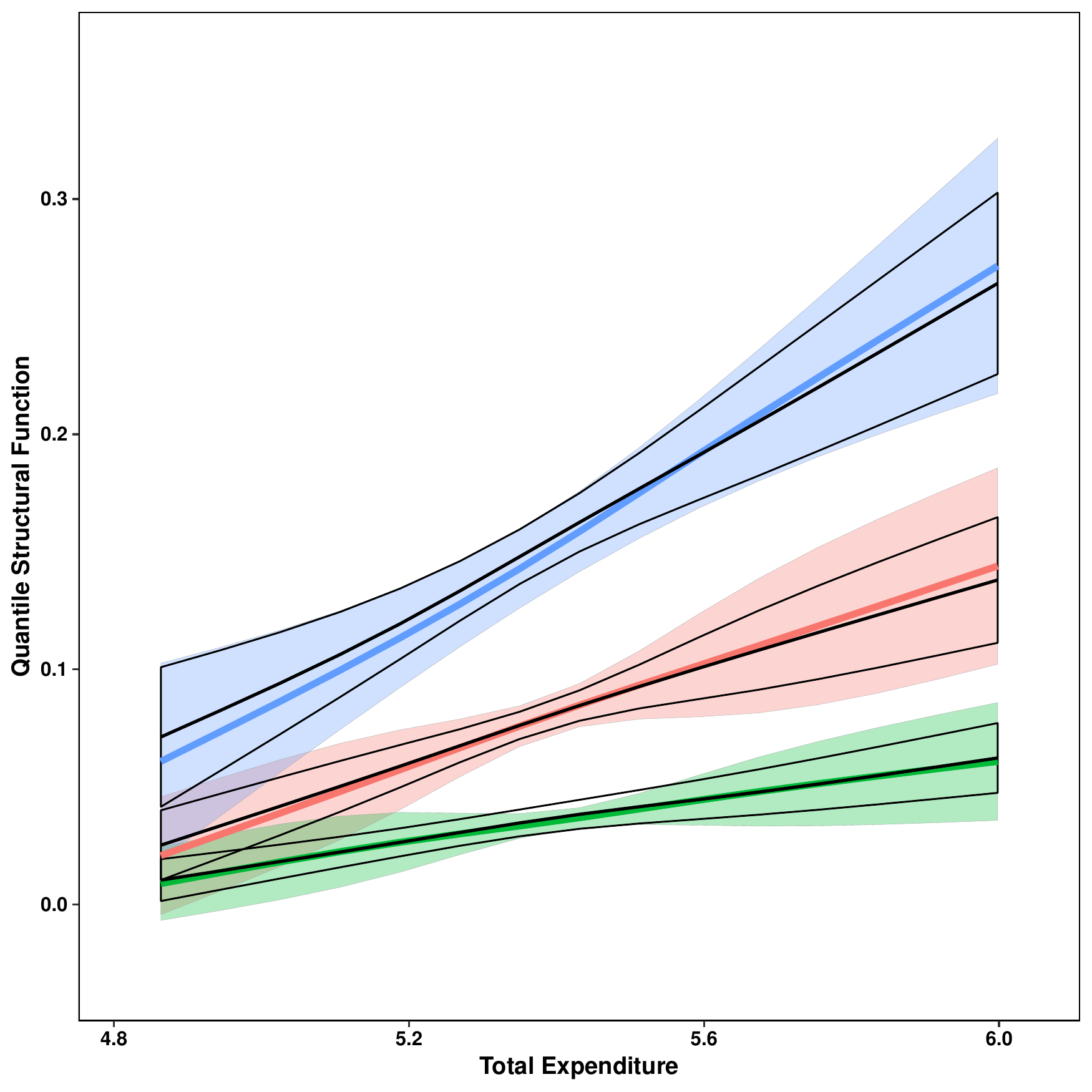} 

}

\subfloat[$M=5$.]{\includegraphics[width=7.5cm,height=6cm]{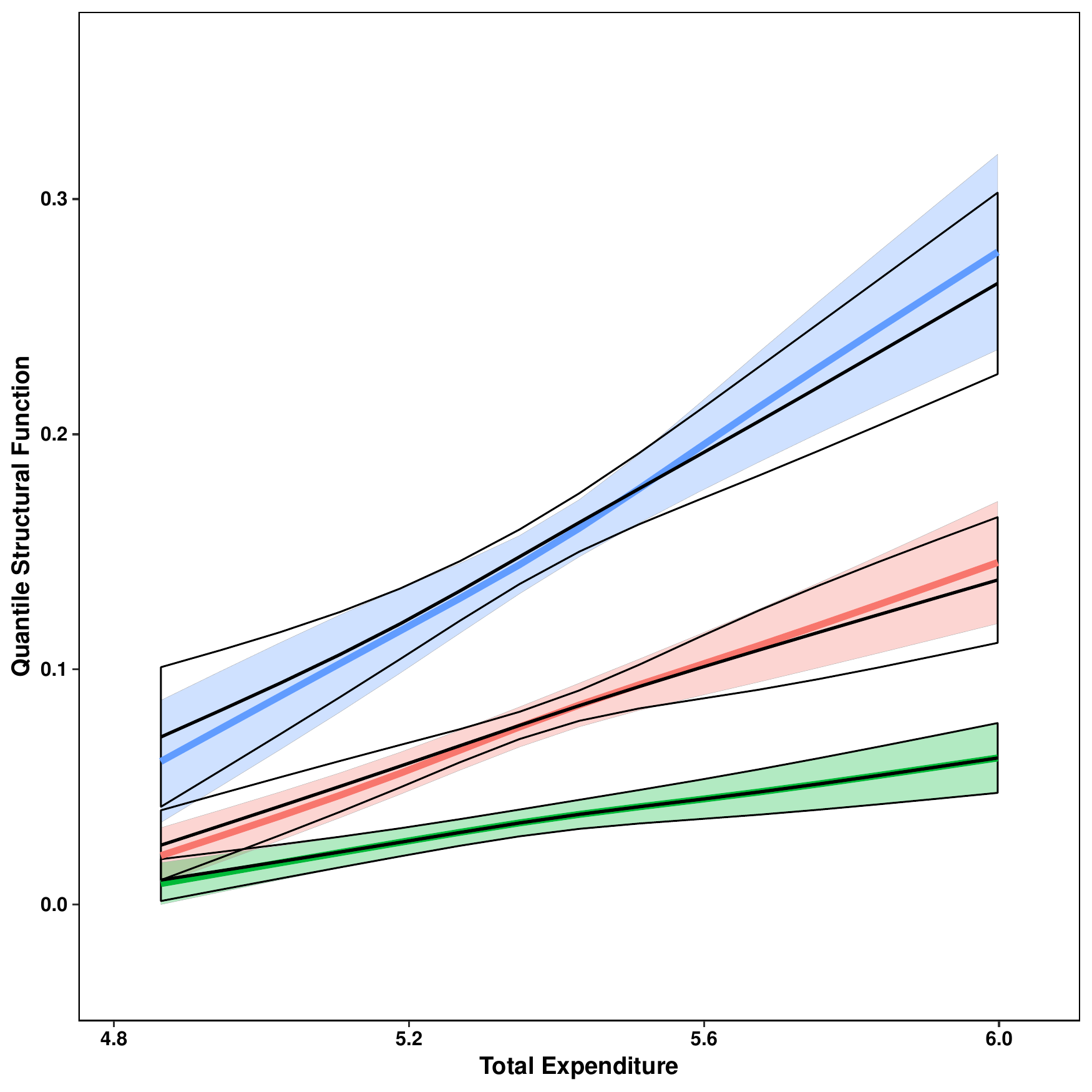}

}\subfloat[$M=15$.]{\includegraphics[width=7.5cm,height=6cm]{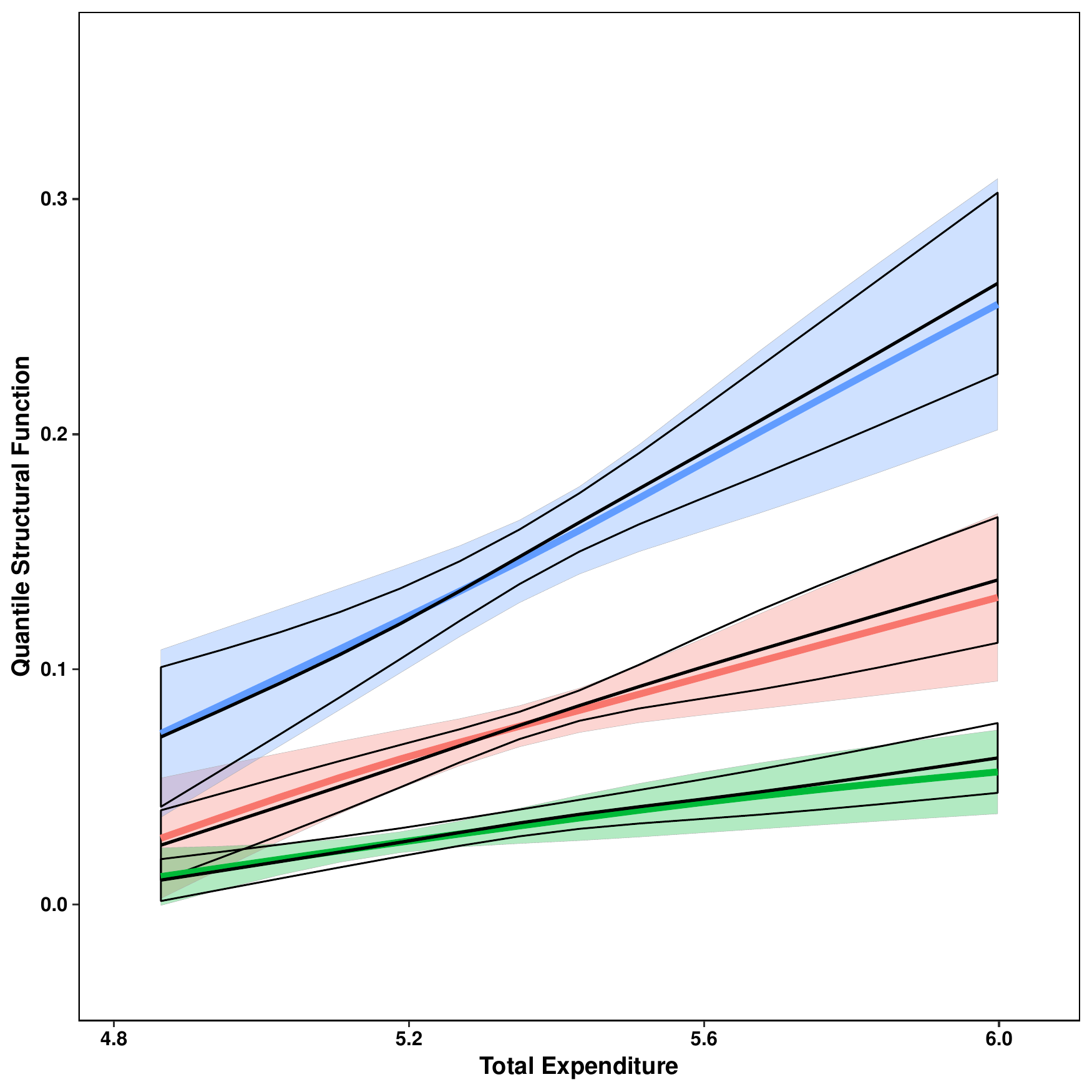}

}

\caption{\textsc{Design} 1. QSF for leisure with discrete instrument $\widetilde{Z}$
(\foreignlanguage{british}{coloured}) and with continuous instrument
$Z^{*}$ (black).}
\label{fig:QSF_Des1_leisure} 
\end{figure}

\begin{figure}[t!]
\subfloat[$M=2$.]{\includegraphics[width=7.5cm,height=6cm]{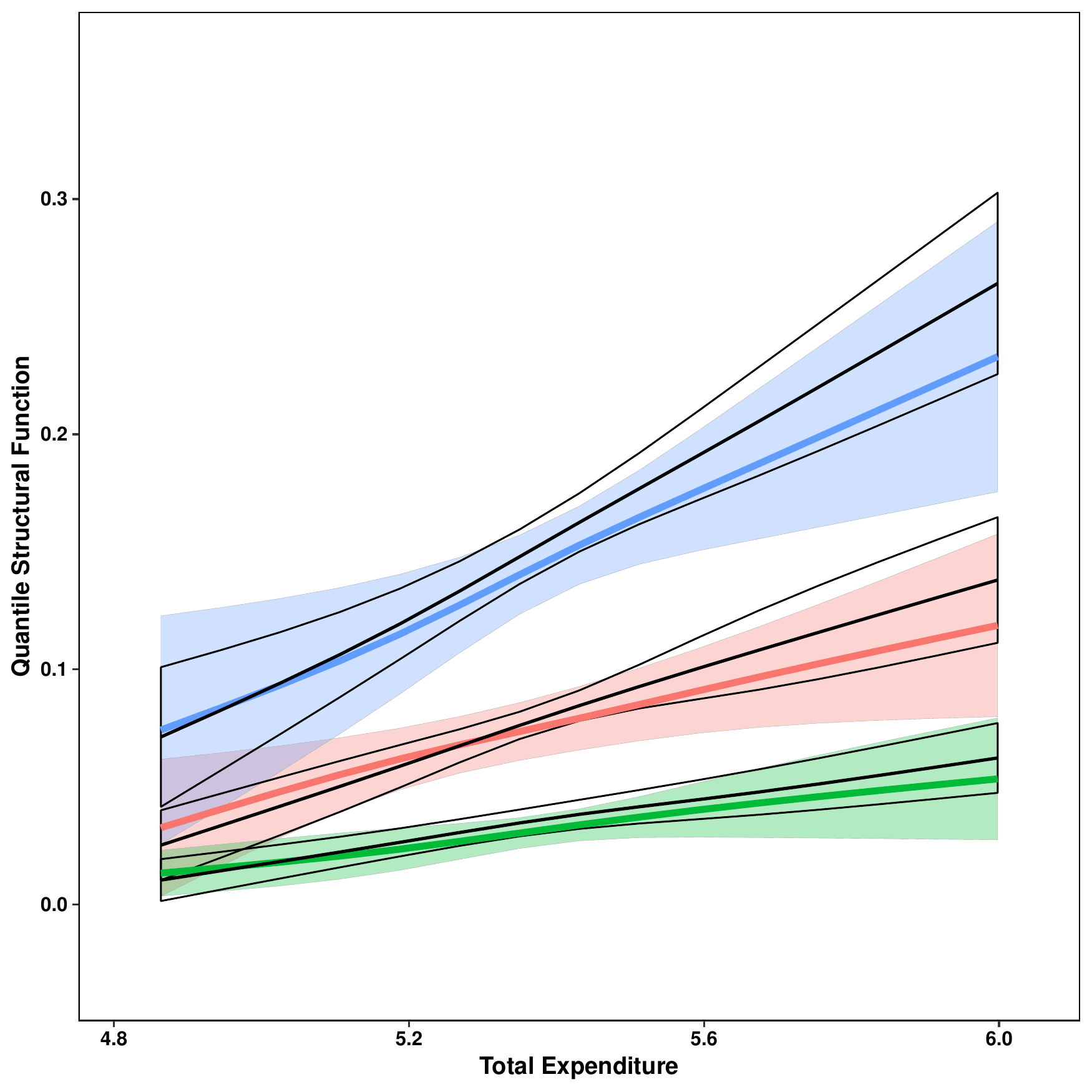}

}\subfloat[$M=3$.]{\includegraphics[width=7.5cm,height=6cm]{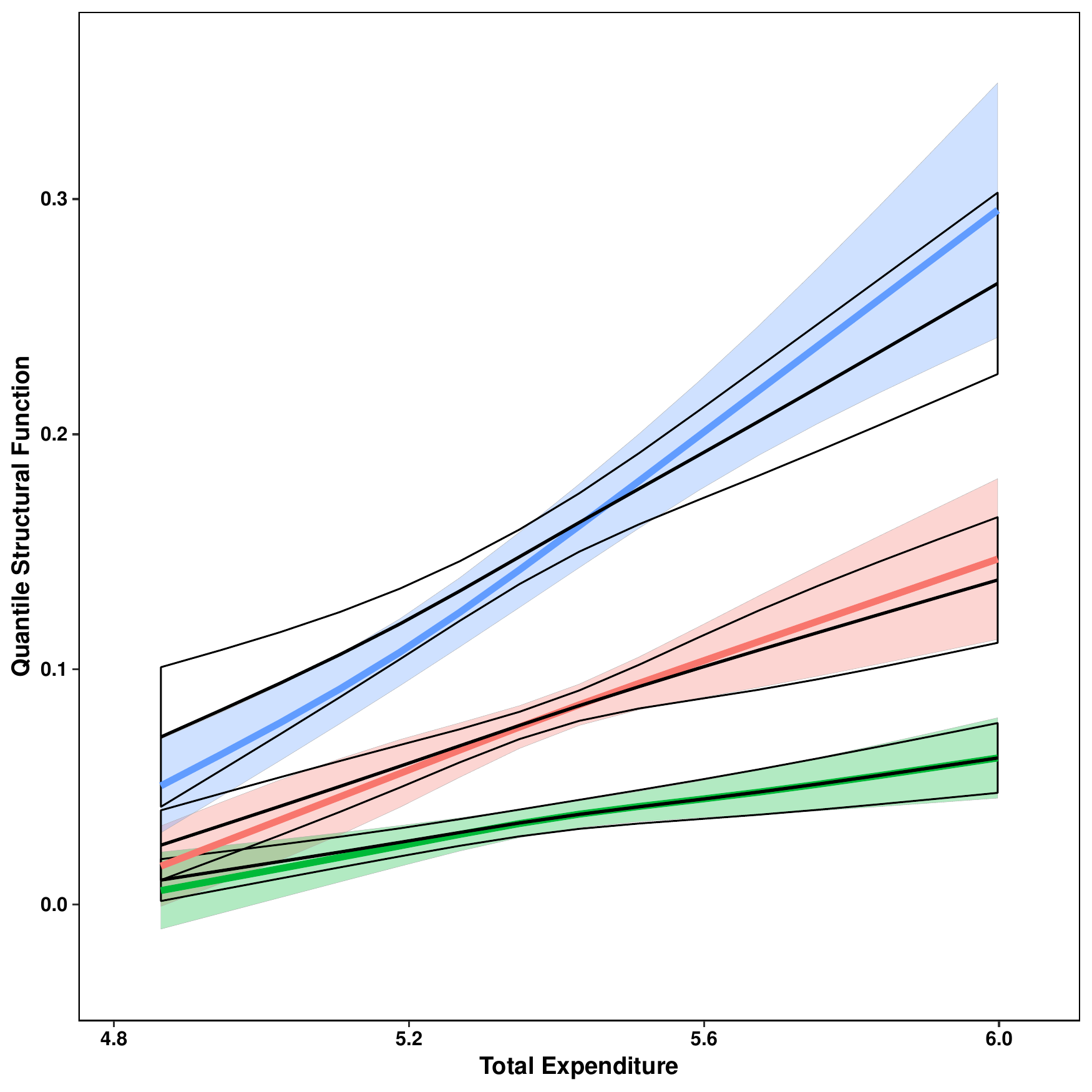}

}

\subfloat[$M=5$.]{\includegraphics[width=7.5cm,height=6cm]{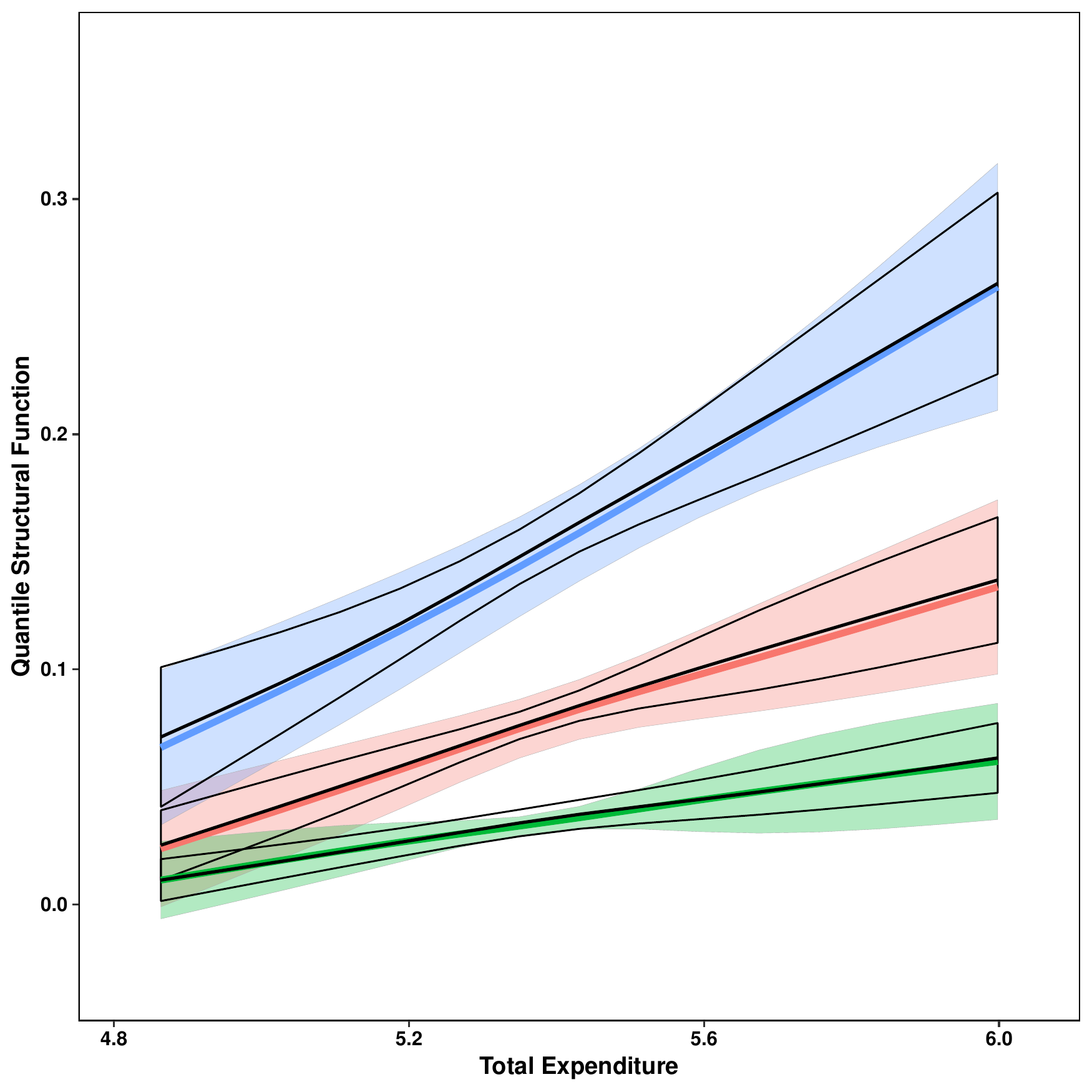}

}\subfloat[$M=15$.]{\includegraphics[width=7.5cm,height=6cm]{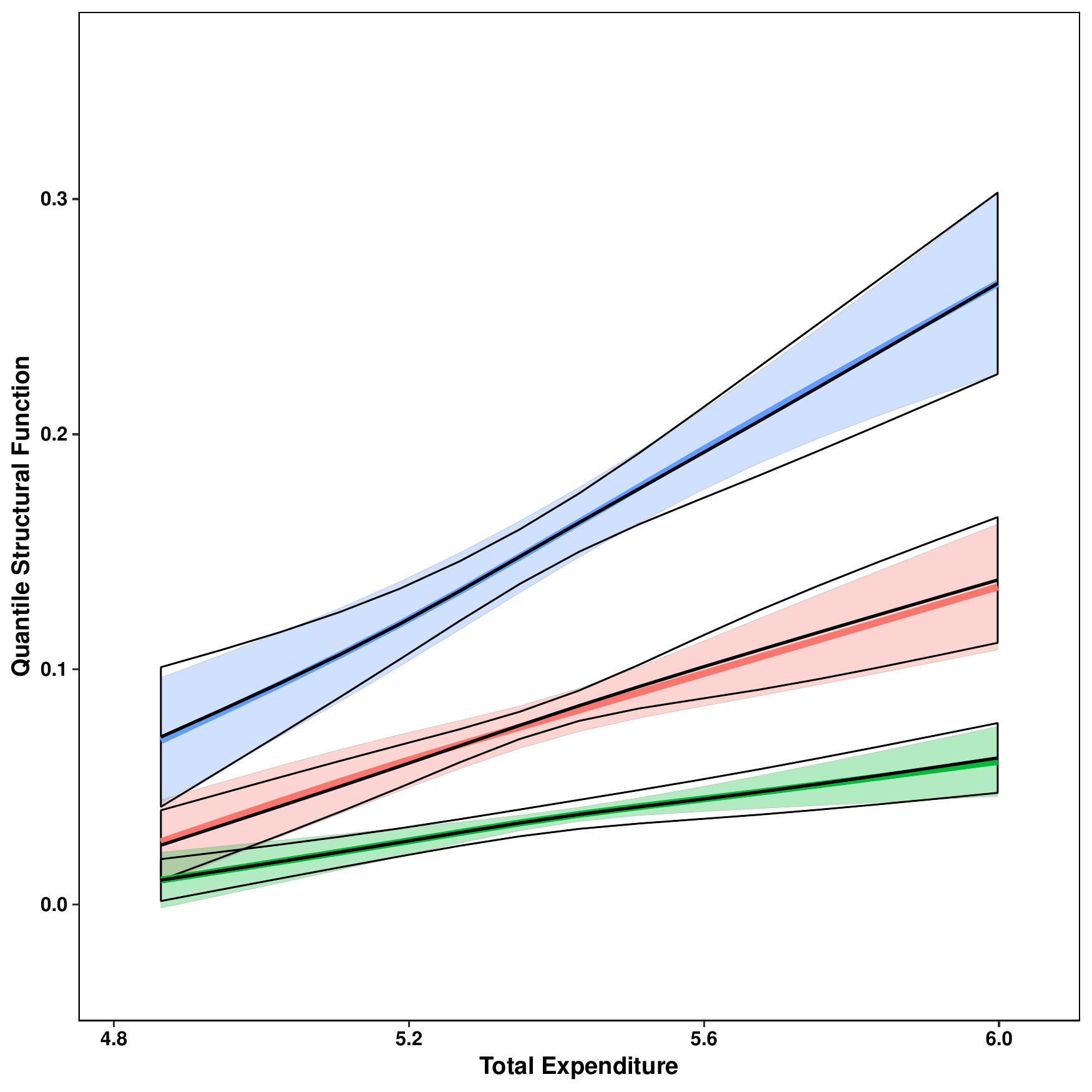}

}

\caption{\textsc{Design} 2. QSF for leisure with discrete instrument $\widetilde{Z}$
(\foreignlanguage{british}{coloured}) and with continuous instrument
$Z^{*}$ (black).}
\label{fig:QSF_Des2_leisure} 
\end{figure}

Figures \ref{fig:QSF_Des1_food} and \ref{fig:QSF_Des2_food} show
the $0.25$, $0.5$ and $0.75$-QSFs for food estimated with each
set of four instruments, respectively, as well as the corresponding
benchmark QSFs estimated using the original continuous instrument
$Z^{*}$. Figures \ref{fig:QSF_Des1_leisure} and \ref{fig:QSF_Des2_leisure}
show the corresponding QSFs for leisure. For comparison purposes the
implementation is exactly as in Chernozhukov et al. (2017). We report
weighted bootstrap 90\%-confidence bands that are uniform over the
support regions of the displayed QSFs,\footnote{All QSFs and uniform confidence bands are obtained over the region
$[\widehat{Q}_{X}(0.1),\widehat{Q}_{X}(0.9)]\times\{0.25,0.5,0.75\}$,
where the interval $[\widehat{Q}_{X}(0.1),\widehat{Q}_{X}(0.9)]$
is approximated by a grid of 5 points $\{\widehat{Q}_{X}(0.1),\widehat{Q}_{X}(0.3),\ldots,\widehat{Q}_{X}(0.9)\}$.
For graphical representation the QSFs are then interpolated by splines
over that interval.} constructed with $250$ bootstrap replications. Our empirical results
show that both discretisation schemes deliver very similar QSF estimates
and confidence bands that capture the main features of the benchmark
QSFs estimated with a continuous instrument. The largest deviations
from the benchmark QSFs occur for $M=2$ and the non uniform \textsc{Design}
2, where the first value of $\widetilde{Z}$ is allocated to $6\%$
of the observations only.

For this dataset the main features of Engel curves for food and leisure
are well captured when estimation is performed with a discrete valued
instrumental variable.\footnote{We have implemented additional robustness checks by estimating the
average and distribution structural functions, as well as nonlinear
specifications of the QSF, when the vector $p(X)$ is augmented with
spline transformations of $X$. Our empirical findings for these objects
are qualitatively similar.} Overall our empirical findings support our identification results
and illustrate the use of discrete instruments for the estimation
of structural functions in %parametric 
triangular systems.

\clearpage{}

\appendix
%dummy comment inserted by tex2lyx to ensure that this paragraph is not empty%dummy comment inserted by tex2lyx to ensure that this paragraph is not empty

\section{Proof of Main Results\label{sec:Proofs}}

\subsection{Proof of Theorem \protect\ref{thm:Theorem1}}
\begin{proof}
\begin{sloppy}Part (i). The proof builds on the proof of Lemma S3
in Spady and Stouli (2018). The matrix $E\left[w(X,V)w(X,V)'\right]$
is of the form 
\begin{align*}
E\left[w(X,V)w(X,V)'\right] & =E\left[\{p(X)\otimes q(V)\}\{p(X)\otimes q(V)\}'\right]\\
 & =E\left[\{p(X)p(X)'\}\otimes\{q(V)q(V)'\}\right]\\
 & =E\left[\begin{array}{cc}
p(X)p(X)' & p(X)p(X)'V\\
p(X)p(X)'V & p(X)p(X)'V^{2}
\end{array}\right].
\end{align*}
Assumption \ref{ass:P(X)posddef}(a) implies that $E[p(X)p(X)']$
is positive definite. Thus $E\left[w(X,V)w(X,V)'\right]$ is positive
definite if and only if the Schur complement of $E[p(X)p(X)']$ in
$E\left[w(X,V)w(X,V)'\right]$ is positive definite (Boyd and Vandenberghe,
2004, Appendix A.5), i.e., if and only if 
\[
\varUpsilon:=E\left[p(X)p(X)'V^{2}\right]-E\left[p(X)p(X)'V\right]E\left[p(X)p(X)'\right]^{-1}E\left[p(X)p(X)'V\right]
\]
satisfies $\textrm{det}(\varUpsilon)>0$.

With 
\[
\Xi=E\left[p(X)p(X)'V\right]E\left[p(X)p(X)'\right]^{-1},
\]
we have that 
\[
\varUpsilon=E\left[\left\{ p(X)V-\Xi p(X)\right\} \left\{ p(X)V-\Xi p(X)\right\} '\right],
\]
a finite positive definite matrix, if and only if for all $\lambda\neq0$
there is no $d$ such that $\Pr[\{\lambda'p(X)\}V=d'\{\Xi p(X)\}]>0$;
this is an application of the Cauchy-Schwarz inequality for matrices
stated in Tripathi (1999).\end{sloppy}

For $\widetilde{\mathcal{X}}\subseteq\mathcal{X}_{V}^{o}$, positive
definiteness of $E[1(X\in\widetilde{\mathcal{X}})p(X)p(X)']$ under
Assumption \ref{ass:P(X)posddef}(a) implies that for all $\lambda\neq0$,
$E[1(X\in\widetilde{\mathcal{X}})\{\lambda'p(X)\}^{2}]>0$, which
implies that for all $\lambda\neq0$, the set $\{x\in\widetilde{\mathcal{X}}\,:\,\lambda'p(x)\neq0\}$
has positive probability. By definition of $\mathcal{V}_{x}$ and
the variance, we have that $\textrm{Var}(V\mid X=x)>0$ for each $x\in\mathcal{X}_{V}^{o}$.
Thus for all $\lambda\neq0$, by $\Xi$ being a constant matrix, there
is no $d$ such that $\Pr[\{\lambda'p(X)\}V=d'\{\Xi p(X)\}]>0$, and
$E\left[w(X,V)w(X,V)'\right]$ is positive definite.

Part (ii). The proof is similar to Part (i). 
\end{proof}

\subsection{Proof of Theorem \protect\ref{thm:Theorem3}}
\begin{proof}
By iterated expectations, $E\left[w(X,V)w(X,V)'\right]$ can be expressed
as 
\[
E\left[w(X,V)w(X,V)'\right]=E\left[\{p(X)p(X)'\}\otimes E[q(V)q(V)'\mid X]\right].
\]
We show that $E\left[w(X,V)w(X,V)'\right]$ is positive definite.
By Assumption \ref{ass:P(X)posddef}(a), there is a positive constant
$B$ such that 
\begin{align*}
E\left[\{p(X)p(X)'\}\otimes E[q(V)q(V)'\mid X]\right] & \geq E\left[1(X\in\widetilde{\mathcal{X}})\{p(X)p(X)'\}\otimes\lambda_{\min}(X)I_{K}\right]\\
 & \geq E\left[1(X\in\widetilde{\mathcal{X}})\{p(X)p(X)'\}\right]\otimes BI_{K},
\end{align*}
where $I_{K}$ is the $K\times K$ identity matrix, and the inequality
means no less than in the usual partial ordering for positive semi-definite
matrices. The conclusion then follows by the matrix following the
last inequality being positive definite by Assumption \ref{ass:P(X)posddef}(a).

Under Assumption \ref{ass:P(X)posddef}(b) the proof is similar upon
using that $E\left[w(X,V)w(X,V)'\right]=E\left[E[p(X)p(X)'\mid V]\otimes\{q(V)q(V)'\}\right]$. 
\end{proof}

\subsection{Proof of Theorem \protect\ref{thm:Theorem4}}
\begin{proof}
By Theorem \ref{thm:Theorem3} the matrix $E\left[w(X,V)w(X,V)'\right]$
exists and is positive definite. The result then follows by Theorem
1 in Chernozhukov et al. (2017). 
\end{proof}
%\subsection{Proof of Proposition \protect\ref{prop:Proposition1}}
%\begin{proof}
%The proof builds on the proof of Theorem 5 in Newey and Stouli (2018).
%For $x\in\mathcal{X}_{Z}^{*}$, by definition of $\mathcal{Z}_{x}$
%we have that $\Pr(Z=z_{m}\mid X=x)\geq\delta>0$ for $m\in\{1,\ldots,|\mathcal{Z}_{x}|\}$,
%and 
%\begin{align*}
%E[q(V)q(V)'\mid X=x]= & \sum_{m=1}^{|\mathcal{Z}_{x}|}\left\{ \left\{ q(F_{X\mid Z}(x\mid z_{m}))q(F_{X\mid Z}(x\mid z_{m}))'\right\} \right.\\
% & \left.\times\Pr(Z=z_{m}\mid X=x)\right\} ,
%\end{align*}
%is a sum of $|\mathcal{R}(x)|\leq|\mathcal{Z}_{x}|$ rank one $K\times K$
%distinct matrices which is singular if $|\mathcal{R}(x)|<K$. If there
%is no set $\widetilde{\mathcal{X}}\subseteq\mathcal{X}$ of positive
%probability such that $\inf_{x\in\widetilde{\mathcal{X}}}|\mathcal{R}(x)|\geq K$,
%the matrix $E[q(V)q(V)'\mid X]$ is then singular with probability
%one, and Assumption \ref{ass:P(X)posddef}(a) cannot hold with $\widetilde{\mathcal{X}}\subseteq\mathcal{X}_{Z}^{*}$,
%by definition of $\mathcal{X}_{Z}^{*}$. Therefore Assumption \ref{ass:P(X)posddef}(a)
%holds with $\widetilde{\mathcal{X}}\subseteq\mathcal{X}_{Z}^{*}$
%only if there is a set $\widetilde{\mathcal{X}}$ of positive probability
%such that $\inf_{x\in\widetilde{\mathcal{X}}}|\mathcal{R}(x)|\geq K$.
%A similar argument shows that $\mathcal{V}_{Z}^{*}$ has positive
%probability only if there is a set $\widetilde{\mathcal{V}}\subseteq(0,1)$
%of positive probability such that $\inf_{v\in\widetilde{\mathcal{V}}}|\mathcal{Q}(v)|\geq J$. 
%\end{proof}

\subsection{Proof of Theorem \protect\ref{thm:Theorem5}}
\begin{proof}
The result follows from the proof of Theorem 1 in Newey and Stouli
(2018). 
\end{proof}

\subsection{Proof of Theorem \protect\ref{thm:Theorem6}}
\begin{proof}
Under Assumption \ref{ass:NPDRModel}(a), $q_{\tau}\left(V\right)$
is identified for each $\tau\in\mathcal{T}$ by Theorem \ref{thm:Theorem5}.
This implies that, for $\mathcal{T}=\mathcal{Y}$, the conditional
CDF $F_{Y\mid XV}(y\mid X,V)=\Gamma(p(X)'q_{y}(V))$ is unique with
probability one for each $y\in\mathcal{Y}$, since $p\left(X\right)$
and $\Gamma$ are known functions. The structural functions are then
identified by (\ref{eq:structure}) in the main text. For $\mathcal{T}=\mathcal{U}$,
when $Y$ is continuous the conditional quantile function $Q_{Y\mid XV}(u\mid X,V)=p(X)'q_{u}(V)$
is unique with probability one for each $u\in\mathcal{U}$. Since
$y\mapsto F_{Y\mid XV}(y\mid XV)$ is the inverse function of $u\mapsto Q_{Y\mid XV}(u\mid X,V)$,
the structural functions are also identified by (\ref{eq:structure})
in the main text. 
\end{proof}

\section{Formal Statement of Remark \protect\ref{rem:Remark2}\label{sec:ProofRemark2}}
\begin{prop}
(i) Let $q(V)=(1,V)^{\prime}$. If Assumption \ref{ass:P(X)posddef}(a)
holds with $\widetilde{\mathcal{X}}\subseteq\mathcal{X}_{V}^{o}$
then it also holds with $\widetilde{\mathcal{X}}\subseteq\mathcal{X}_{V}^{*}$.
(ii) Let $p(X)=(1,X)^{\prime}$. If Assumption \ref{ass:P(X)posddef}(b)
holds with $\widetilde{\mathcal{V}}\subseteq\mathcal{V}_{X}^{o}$
then it also holds with $\widetilde{\mathcal{V}}\subseteq\mathcal{V}_{X}^{*}$.\label{prop:Proposition2} 
\end{prop}
\begin{proof}
Each $x\in\mathcal{X}_{V}^{o}$ satisfies $|\mathcal{V}_{x}|\geq2$,
which by the definitions of $\mathcal{V}_{x}$ and the variance implies
that $\textrm{Var}(V\mid X=x)\geq B>0$. For $q(V)=(1,V)'$, the smallest
eigenvalue of $E[q(V)q(V)'\mid X=x]$ is then bounded away from zero
for each $x\in\mathcal{X}_{V}^{o}$, by Lemma \ref{lem:Lemma1} below.
Therefore each $x\in\mathcal{X}_{V}^{o}$ also satisfies $x\in\mathcal{X}_{V}^{*}$,
so that $\mathcal{X}_{V}^{o}\subseteq\mathcal{X}_{V}^{*}$. The result
for Part (i) follows, and the proof for Part (ii) is similar. 
\end{proof}
\begin{lem}
For a set of random variables $\{X(t)\}_{t\in\overline{\mathcal{T}}}$
such that $E[X(t)^{2}]\leq C$ and $\text{Var}(X(t))\geq B>0$, the
smallest eigenvalue of 
\[
\Sigma(t)=E\left[\left(\begin{array}{c}
1\\
X(t)
\end{array}\right)\left(\begin{array}{cc}
1 & X(t)\end{array}\right)\right]
\]
is bounded away from zero.\label{lem:Lemma1} 
\end{lem}
\begin{proof}
We have $\textrm{det}(\Sigma(t))=\textrm{Var}(X(t))=\lambda_{\max}(t)\lambda_{\min}(t)$,
where $\lambda_{\max}(t)$ and $\lambda_{\min}(t)$ are the largest
and smallest eigenvalues of $\Sigma(t)$, respectively. Note that,
for some positive constant $\widetilde{C}<\infty$ and all $t\in\overline{\mathcal{T}}$,
\[
\lambda_{\max}(t)=\sup_{\lambda:||\lambda||=1}\lambda'\Sigma(t)\lambda\leq||\lambda||^{2}||\Sigma(t)||\leq||\Sigma(t)||\leq\widetilde{C}
\]
by $E[X(t)^{2}]$ bounded. Therefore 
\[
\lambda_{\min}(t)=\frac{\textrm{Var}(X(t))}{\lambda_{\max}(t)}\geq\frac{\textrm{Var}(X(t))}{\widetilde{C}}\geq\frac{B}{\widetilde{C}},
\]
and the result follows. 
\end{proof}

\section{Estimation of Structural Functions\label{sec:Quantile-Regression-Estimation}}

Here we give a summary of the key steps in the implementation of the
quantile regression-based estimators for structural functions proposed
by Chernozhukov et al. (2017). A detailed description and implementation
algorithms for estimation and the weighted bootstrap procedures are
given in Chernozhukov et al. (2017).

The estimators implemented in the empirical application have three
main stages. In the first stage, we estimate the control variable,
$\{\widehat{V}_{i}\}_{i=1}^{n}$. In the second stage, we estimate
the distribution CRF, $\widehat{F}_{Y\mid XZ_{1}V}(y\mid x,z_{1},v)$.
In the third and final stage, estimators $\widehat{G}(y,x)$, $\widehat{Q}(\tau,x)$
and $\widehat{\mu}(x)$ of the distribution, quantile and average
structural functions, respectively, are obtained.

\textbf{First stage.} {[}Control function estimation{]} Denoting the
usual check function by $\rho_{v}(z)=(v-1(z<0))z$, the quantile regression
estimator of $F_{X\mid Z}$ is, for $(x,z)\in\mathcal{X}\mathcal{Z}$,
\begin{eqnarray}
\widehat{F}_{X\mid Z}(x\mid z) & = & \epsilon+\int_{\epsilon}^{1-\epsilon}1\left\{ \widehat{\pi}(v)^{\prime}[s(\widetilde{z})\otimes r(z_{1})]\leq x\right\} dv,\label{eq:FhatXZ}\\
\widehat{\pi}(v) & \in & \arg\min_{\pi\in\mathbb{R}^{\text{dim}(Z)}}\sum_{i=1}^{n}\rho_{v}(X_{i}-\pi^{\prime}[s(\widetilde{Z}_{i})\otimes r(Z_{1i})]),\label{eq:pihat}
\end{eqnarray}
for some small constant $\epsilon>0$. The adjustment in the limits
of the integral in (\ref{eq:FhatXZ}) avoids tail estimation of quantiles\footnote{Chernozhukov et al. (2013) provide conditions under which this adjustment
does not introduce bias. }. In practice, for $\epsilon$ in $(0,0.5)$ (e.g., $\epsilon=0.01$)
and a fine mesh of $T$ values $\{\epsilon=v_{1}<\cdots<v_{T}=1-\epsilon\}$,
estimate $\{\widehat{\pi}(v_{t})\}_{t=1}^{T}$ by solving (\ref{eq:pihat}).
Obtain the control function estimator $\widehat{F}_{X\mid Z}(X_{i}\mid Z_{i})$
as in (\ref{eq:FhatXZ}), and set $\widehat{V}_{i}=\widehat{F}_{X\mid Z}(X_{i}\mid Z_{i})$,
for $i\in\{1,\ldots,n\}$.\medskip{}

\textbf{Second stage.} {[}Distribution CRF estimation{]} The quantile
regression estimator of $F_{Y\mid XZ_{1}V}$ is, for $(y,x,z_{1},v)\in\mathcal{Y}\mathcal{X}\mathcal{Z}_{1}\mathcal{V}$,
\begin{eqnarray}
\widehat{F}_{Y\mid XZ_{1}V}(y\mid x,z_{1},v) & = & \epsilon+\int_{\epsilon}^{1-\epsilon}1\left\{ \widehat{\beta}(u)^{\prime}w(x,z_{1},v)\leq y\right\} du,\label{eq:FhatYXV}\\
\widehat{\beta}(u) & \in & \arg\min_{\beta\in\mathbb{R}^{\text{dim}(w(X,Z_{1},V))}}\sum_{i=1}^{n}\rho_{u}(Y_{i}-\beta^{\prime}w(X_{i},Z_{1i},\widehat{V}_{i})),\label{eq:betahat}
\end{eqnarray}
In practice, for $\epsilon$ in $(0,0.5)$ (e.g., $\epsilon=0.01$)
and a fine mesh of $T$ values $\{\epsilon=u_{1},\ldots,u_{T}=1-\epsilon\}$,
estimate $\{\widehat{\beta}(u_{t})\}_{t=1}^{T}$ by solving (\ref{eq:betahat}).
Obtain the distribution CRF estimator $\widehat{F}_{Y\mid XZ_{1}V}(y\mid x,Z_{1i},\widehat{V}_{i})$
as in (\ref{eq:FhatYXV}).\medskip{}

\textbf{Third stage.} {[}Structural functions estimation{]} Let $\mathcal{Y}^{+}=\mathcal{Y}\cap\lbrack0,\infty)$
and $\mathcal{Y}^{-}=\mathcal{Y}\cap(-\infty,0)$. Given estimates
$(\{\widehat{V}_{i}\}_{i=1}^{n},\widehat{F}_{Y\mid XZ_{1}V})$, the
estimator for the distribution structural function takes the form
\[
\widehat{G}(y,x)=\frac{1}{n}\sum_{i=1}^{n}\widehat{F}_{Y\mid XZ_{1}V}(y\mid x,Z_{1i},\widehat{V}_{i}).
\]
Given the distribution structural function estimate, the QSF estimator
is defined as 
\begin{equation}
\widehat{Q}(p,x)=\int_{\mathcal{Y}^{+}}1\{\widehat{G}(y,x)\leq p\}dy-\int_{\mathcal{Y}^{-}}1\{\widehat{G}(y,x)\geq p\}dy,\label{eq:QSFhat}
\end{equation}
and the average structural function estimator as 
\begin{equation}
\widehat{\mu}(x)=\int_{\mathcal{Y}^{+}}[1-\widehat{G}(y,x)]\nu(dy)-\int_{\mathcal{Y}^{-}}\widehat{G}(y,x)\nu(dy),\label{eq:ASFhat}
\end{equation}
where $\nu$ is either the counting measure when $\mathcal{Y}$ is
countable or the Lebesgue measure otherwise. When the set $\mathcal{Y}$
is uncountable and bounded, we approximate the previous integrals
by sums over a fine mesh of equidistant points $\mathcal{Y}_{S}:=\{\inf[y\in\mathcal{Y}]=y_{1}<\cdots<y_{S}=\sup[y\in\mathcal{Y}]\}$
with mesh width $\delta$ such that $\delta\sqrt{n}\to0$. For example,
(\ref{eq:QSFhat}) and (\ref{eq:ASFhat}) are approximated by

\[
\widehat{G}(y,x)=\frac{1}{n}\sum_{i=1}^{n}\widehat{F}_{Y\mid XZ_{1}V}(y\mid x,Z_{1i},\widehat{V}_{i}),
\]
\[
\widehat{Q}_{S}(p,x)=\delta\sum_{s=1}^{S}\left[1(y_{s}\geq0)-1\{\widehat{G}(y_{s},x)\geq p\}\right],\,\widehat{\mu}_{S}(x)=\delta\sum_{s=1}^{S}\left[1(y_{s}\geq0)-\widehat{G}(y_{s},x)\right].
\]
%\medskip{}

The choices of $\epsilon$ and %the size of the grids 
$T$ can differ across stages. In the empirical application we set
$\epsilon=0.01$ and $T=599$ throughout. For the third stage, we
approximate the integrals in (\ref{eq:QSFhat})-(\ref{eq:ASFhat})
using $S=599$ points. Overall, for this %empirical 
application %we have found that 
the estimates are not very sensitive to $T$, and are also robust
to varying values of $\epsilon$ and $S$.


\begin{thebibliography}{99}
\bibitem{BC2001} Belloni, A., and V. Chernozhukov, 2011, \newblock
{\ensuremath{\ell}1-penalized quantile regression in high-dimensional
sparse models}. \newblock Annals of Statistics 39, pp. 82-130.

\bibitem{BCFH2017} Belloni, A., and V. Chernozhukov, 2011, Fernandez-Val,
I., and C. Hansen, 2017, \newblock {Program evaluation and causal
inference with high-dimensional data}. \newblock Econometrica 85,
pp. 233-298.

\bibitem[{Boyd & Vandenberghe(2004)}]{BV:2004}Boyd, S. P. and L.
Vandenberghe, 2004, \newblock {Convex optimization}. \newblock
Cambridge University Press, Cambridge.

\bibitem{Blundell Powell 2003}Blundell, R., and J. L. Powell, 2003,
\newblock {Endogeneity in nonparametric and semiparametric regression
models}. \newblock Econometric society monographs 36, pp. 312-357.

\bibitem{Blundell Powell 2004} Blundell, R., and J. L. Powell, 2004,
\newblock {Endogeneity in semiparametric binary response models}.
\newblock The Review of Economic Studies 71, pp. 655-679.

\bibitem{Bl Chen Kris 2007} Blundell, R., Chen, X., and D. Kristensen,
2007, \newblock {Semi-nonparametric IV estimation of shape-invariant
Engel curves}. \newblock Econometrica 75, pp. 1613-1669.

\bibitem{Chern Fern Kow}Chernozhukov, V., Fernandez-Val, I., and
A. Kowalski, 2015, \newblock {Quantile regression with censoring
and endogeneity}. \newblock Journal of Econometrics 186, pp. 201-221.

\bibitem{CFVM2013}Chernozhukov, V., Fernandez-Val, I., and B. Melly,
2013, \newblock {Inference on counterfactual distributions}. \newblock
Econometrica 81, 2205-2268.

\bibitem{CFNSV:2017} Chernozhukov, V., Fernandez-Val, I. Newey, W.,
Stouli, S. and F. Vella, 2017, \newblock {Semiparametric estimation
of structural functions in nonseparable triangular models}. \newblock
\textit{eprint arXiv:1711.02184}.

\bibitem{Fev Hault 2015} D'Haultfœuille, X. and P. Février, 2015,
\newblock {Identification of nonseparable triangular models with
discrete instruments}. \newblock Econometrica 83, pp. 1199-1210.

\bibitem{FVV:2018} Fernandez-Val, I. Van Vuuren, A. and F. Vella,
2017, \newblock {Nonseparable sample selection models with censored
selection rules}. \newblock \textit{eprint arXiv:1801.08961}.

\bibitem{Flo Heck Meg Vytl 2008} Florens, J. P., Heckman, J. J.,
Meghir, C. and E. Vytlacil, 2008, \newblock {Identification of treatment
effects using control functions in models with continuous, endogenous
treatment and heterogeneous effects}. \newblock Econometrica 76,
pp. 1191-1206.

\bibitem{HN2016} Hausman, J. A. and W. K. Newey, 2016, \newblock
{Individual heterogeneity and average welfare}. \newblock Econometrica
84, pp.1225-1248.

\bibitem{HW1978} Hausman, J. A. and D. Wise, 1978, \newblock{ A
conditional probit model for qualitative choice: discrete decisions
recognizing interdependence and heterogeneous preferences}. \newblock
Econometrica 46, pp. 403-426.

%\bibitem{HM2007} Hoderlein, S. and E. Mammen, 2007, \newblock {Identification
%of marginal effects in nonseparable models without monotonicity}.
%\newblock Econometrica 75, pp. 1513-1518.
\bibitem[{Imbens(2000)}]{Imbens 2000} Imbens, G. W., 2000, \newblock
{The role of the propensity score in estimating dose-response functions}.
\newblock Biometrika 87(3), pp. 706-710.

\bibitem{Imbens Newey 2009} Imbens, G. and W. K. Newey, 2009, \newblock
{Identification and estimation of triangular simultaneous equations
models without additivity}. \newblock Econometrica 77, pp. 1481-1512.

\bibitem[{Kitamura and Stoye(2018)}]{KS2018}Kitamura, Y. and J. Stoye,
2018, \newblock {Nonparametric analysis of random utility models}.
\newblock Econometrica 86, pp.1883-1909.

\bibitem[{Masten and Torgovitsky(2016)}]{MT:2016}Masten, M. and A.
Torgovitsky, 2016, \newblock {Identification of instrumental variable
correlated random coefficients models}. \newblock Review of Economics
and Statistics 98, pp. 1001\textendash 1005.

\bibitem{McF1973} McFadden, D., 1973, \newblock {Conditional logit
analysis of qualitative choice behavior}, \newblock in: P. Zarambka
(Ed.), Frontiers in econometrics. New York: Academic Press.

\bibitem{NMcF:1994} Newey, W.K. and D. McFadden, 1994, \newblock
{Large sample estimation and hypothesis testing}, \newblock in:
Engle, R. and D. McFadden (Eds.), Handbook of econometrics. Elsevier,
Berlin, pp. 2111-2245.

\bibitem{NS:2018} Newey, W. K. and S. Stouli, 2018, \newblock {Heterogenous
coefficients, discrete instruments, and identification of treatment
effects}. \newblock \textit{eprint arXiv:1811.09837.}

\bibitem{key-9} Rosenbaum, P. R. and D. B. Rubin, 1983, \newblock
{The central role of the propensity score in observational studies
for causal effects}. \newblock Biometrika 70, pp.41-55.

\bibitem{SS:2018} Spady, R. H. and S. Stouli, 2018, \newblock {Dual
regression}. \newblock Biometrika 105, pp. 1-18.

\bibitem{key-20}Tibshirani, R., 1996, \newblock {Regression shrinkage
and selection via the Lasso}. \newblock Journal of the Royal Statistical
Society: Series B 58, 267-288.

\bibitem{Torgo 2015} Torgovitsky, A., 2015, \newblock{Identification
of nonseparable models using instruments with small support}. \newblock
Econometrica 83, pp. 1185-1197.

\bibitem[{Wooldridge(2004)}]{Wool:2004}Tripathi, G., 1999, \newblock
{A matrix extension of the Cauchy-Schwarz inequality}. \newblock
Economics Letters 63, pp. 1-3.

\bibitem{key-16}Wooldridge, J. M\textsc{.}, 2004, \newblock {Estimating
average partial effects under conditional moment independence assumptions}.
\newblock Cemmap working paper CWP03/04. 
\end{thebibliography}
\end{document}